\documentclass[lettersize,journal]{IEEEtran}

\usepackage{enumerate}
\usepackage[shortlabels]{enumitem}
\usepackage{booktabs} 
\usepackage{amssymb}
\usepackage{amsfonts}
\usepackage{mathrsfs}
\usepackage{amsmath} 
\usepackage{amsthm}
\usepackage{mathtools}

\DeclareMathOperator{\diag}{diag}

\newcommand{\coded}[1]{\widetilde{#1}}

\newcommand{\ceil}[1]{ {\lceil#1\rceil}}

\usepackage{algorithm}
\usepackage[noend]{algpseudocode}
\algrenewcommand\algorithmiccomment[1]{\textcolor{lightgray}{\hfill // #1}}

\usepackage{pgfplots}
\pgfplotsset{compat=1.18}
\usepackage{graphicx}

\usepackage{hyperref}
\hypersetup{
 colorlinks,
 linkcolor={blue!100!black},
 citecolor={blue!100!black},
 urlcolor={blue!80!black} 
}

\usepackage[style=ieee,
            backend=biber,
            doi=false,
            ]{biblatex}

\usepackage{subcaption}
\usepackage{wrapfig}

\newtheorem{corollary}{Corollary}

\theoremstyle{definition}
\newtheorem{definition}{Definition}
\newtheorem{remark}{Remark}
\newtheorem{theorem}{Theorem}
\newtheorem{lemma}{Lemma}

\numberwithin{rc}{section}

\usepackage{xcolor}
\usepackage{soul} 

\newcommand{\bbF}{\mathbb{F}}

\newcommand{\bfg}{\mathbf{g}}

\newcommand{\bfs}{\mathbf{s}}

\newcommand{\bfw}{\mathbf{w}}
\newcommand{\bfx}{\mathbf{x}}

\newcommand{\bfA}{\mathbf{A}}
\newcommand{\bfB}{\mathbf{B}}
\newcommand{\bfC}{\mathbf{C}}
\newcommand{\bfD}{\mathbf{D}}

\newcommand{\bfG}{\mathbf{G}}

\newcommand{\bfI}{\mathbf{I}}

\newcommand{\bfL}{\mathbf{L}}

\newcommand{\bfP}{\mathbf{P}}
\newcommand{\bfQ}{\mathbf{Q}}

\newcommand{\bfV}{\mathbf{V}}

\newcommand{\cR}{\mathcal{R}}
\newcommand{\cS}{\mathcal{S}}
\newcommand{\cT}{\mathcal{T}}

\newcommand{\rmP}{\mathrm{P}}

\newcommand{\rmS}{\mathrm{S}}
\newcommand{\rmT}{\mathrm{T}}

\begin{document}

\title{On the Encoding Process in Decentralized Systems}

\author{Canran Wang,~\IEEEmembership{Member,~IEEE,}
and Netanel Raviv,~\IEEEmembership{Senior Member,~IEEE}
\IEEEcompsocitemizethanks{
\IEEEcompsocthanksitem The authors are with the Department of Computer Science and Engineering, Washington University in St. Louis, St. Louis, MO 63103 USA. Parts of this work have previously appeared in [DOI:10.1109/ITW54588.2022.9965873].}
}

\maketitle

\begin{abstract}
    We consider the problem of encoding information in a system of~$N=K+R$ processors that operate in a decentralized manner, i.e., without a central processor which orchestrates the operation.
    The system involves~$K$ source processors, each holding some data modeled as a vector over a finite field.
    The remaining~$R$ processors are sinks, and each of which requires a linear combination of all data vectors.
    These linear combinations are distinct from one sink to another, and are specified by a generator matrix of a systematic linear code.

    To capture the communication cost of decentralized encoding, we adopt a linear network model in which the process proceeds in consecutive communication rounds.
    In every round, every processor sends and receives one message through each one of its~$p$ ports.
    Moreover, inspired by network coding literature, we allow processors to transfer linear combinations of their own data and previously received data.
    We propose a framework that addresses the problem on two levels.
    On the \emph{universal} level, we provide a solution to the decentralized encoding problem for~\emph{any} possible linear code.
    On the~\emph{specific} level, we further optimize our solution towards systematic Reed-Solomon codes, as well as their variant, Lagrange codes, for their prevalent use in coded storage and computation systems.
    Our solutions are based on a newly-defined collective communication operation called~\emph{all-to-all encode}.
    \end{abstract}
\section{Introduction}
The problem of decentralized encoding emerges during the encoding phase in coded decentralized systems, including blockchain networks~\cite{li2020polyshard,wang2022breaking,khooshemehr2021discrepancy}, sensor networks~\cite{li2019decode}, and the internet of things~\cite{chen2018privacy}.
In such systems, raw data is generated independently in distributed source nodes, and then linearly encoded and delivered to sink (or parity) nodes via direct node-to-node communication links, without a central authority which orchestrates the operation. 
To study the communication cost in this setting, we focus on the decentralized encoding problem, in which the role of processors can be captured by nodes in a distributed storage/computation system, devices in an IoT system, participants in a blockchain network, and more.

\begin{definition}\label{definition:decentralized-encoding}
\textbf{(Decentralized encoding)}
Consider a system with two disjoint sets of processors~$\cS=\{\rmS_0,\ldots,\rmS_{K-1}\}$ and~$\cT=\{\rmT_0,\ldots,\rmT_{R-1}\}$. 
Every source processor~$\rmS_k\in\cS$ initially possesses some data vector~$\bfx_k\in\bbF_q^W$, where~$\bbF_q$ is a finite field with~$q$ elements, and~$W$ denotes the vector length.
After decentralized encoding, every sink (or parity) processor~$\rmT_r\in\cT$ obtains coded data~$\coded{\bfx}_r\in\bbF_q^W$.
For a generator matrix~$\bfG=\left[
\begin{array}{c|c} 
  \bfI & \bfA 
\end{array} 
\right]\in\bbF_q^{K\times N}$ of some systematic linear code that is known a priori to all processors, where~$\bfI$ is the identity matrix, the coded data satisfies
$$
\left(\coded{\bfx}_0^\intercal,\ldots,\coded{\bfx}_{R-1}^\intercal\right) = \left({\bfx_0}^\intercal,\ldots,{\bfx_{K-1}^\intercal}\right)\cdot \bfA.
$$
\end{definition}
That is, decentralized encoding considers a scenario in which 
every source processor~$\rmS_k\in\cS$ has its own data, and every sink processor~$\rmT_r\in\cT$ requires a linear combination of all data.
The linear combinations are determined by matrix~$\bfA$, i.e., the non-systematic part of the generator matrix~$\bfG$ of some systematic linear code.
Clearly, the properties of the linear code determine the decentralized encoding process.

As a motivating example, consider a local sensor network of~$N$ nodes, $K$~of which are thermometers (sources), and each thermometer holds an independent temperature reading modeled as a finite field element.
The temperature data is to be encoded across the network via device-to-device communication links, so that all temperature readings can be obtained (decoded) from any~$K$ out of the~$N$ nodes; such decoding guarantees can be made using the well-known Reed-Solomon code, which is discussed at length in the sequel.
As another example, consider a distributed storage system with~$N$ geographically separated nodes which store a large database that is locally collected at some~$K$ of them.
Regenerating codes~\cite{dimakis2010network,ramkumar2022codes}, which are widely used in such settings to maintain data availability in cases of node failures, are a special use case of our framework.

Yet another example is the recent \emph{coded computing} literature~\cite{li2020coded}, which studies distributed computation systems with~$N$ nodes which hold data that ought to undergo some polynomial computation (e.g., matrix multiplication or hash function evaluation).
Our work can be used to adapt coded computing systems (e.g., Lagrange code~\cite{yu2019lagrange}) to the decentralized setting, where~$K$ of the nodes generate data locally, and coded computing is to be carried out in a master-less fashion.

Inspired by the extensive network-coding literature~\cite{ahlswede2000network,li2003linear}, we allow processors to transfer linear combinations of their own data and previously received messages to complete the decentralized encoding process.
As a result, any solution to this problem contains two separate components, \emph{scheduling} and \emph{coding scheme}. 
The former determines which processor communicates with which other processors, and the latter determines the coefficients in the linear combinations that processors transmit to one another.

Motivated by this distinction, we address the decentralized encoding problem on two levels, the \emph{universal} and the \emph{specific}. 
On the universal level, we seek a scheduling by which the decentralized encoding problem with \emph{any} linear code could be solved by only varying the coding scheme, i.e., the coefficients in the transmitted packets throughout the algorithm.
That is, a \emph{universal algorithm} is a series of instructions which indicate which processor communicates with which other processors as the algorithm proceeds, alongside a mechanism that for any given~$\bfA$ determines the coefficients that are used by each processor in order to linearly combine previously received data in each transmission.
Universal algorithms are important in cases where the scheduling must be determined prior to knowing~$\bfA$, as well as in cases where the same scheduling is applied in consecutive computation of several different matrices. 
As demonstrated in the sequel, universal algorithms also play an important role as a sub-routine in specific ones.

On the specific level, we seek both scheduling and coding scheme that are uniquely tailored towards a specific linear code of interest. 
Clearly, such specific algorithms are important only if they outperform universal ones, since by definition, every universal algorithm subsumes a specific algorithm for all systematic linear codes.
In this paper, we are particularly interested in Reed-Solomon codes and Lagrange codes~\cite{yu2019lagrange,raviv2019private,wang2022breaking}, for their prevalent use in distributed storage and computation systems.

\begin{remark}
    We assume that all processors have prior knowledge of the system design, including the network structure and the identity of the matrix~$\bfA$. When the matrix is known, the length of each round is also predetermined for every processor.
    By decentralization, we mean that participants operate independently, without a central coordinating authority.
    Importantly, in both the specific and universal settings, neither the scheduling nor the coding scheme depends on the input data~$\bfx_k$; rather, they are determined exclusively by the matrix~$\bfA$, which is known in advance to all processors.
\end{remark}

\subsection*{Communication model}

We assume a fully-connected, multi-port and homogeneous network which allows bidirectional communication.
First, the full connectivity enables direct message delivery between every possible pair of processors.
Second, each processor can simultaneously send a message to one processor, and receive a message from a (potentially different) processor, through each one of its~$p$ ports.
The data which initially resides in each source processor consists of bits, which are grouped together to represent an element in~$\bbF_q$.

We wish to evaluate our algorithms in terms of their running times.
Previous works in distributed storage systems~\cite{dimakis2010network,rashmi2011optimal,tamo2012zigzag,papailiopoulos2013repair} consider~\emph{bandwidth} (i.e., total number of bits communicated) as a major communication metric.
However, as message deliveries may happen in parallel, the (total) bandwidth may not faithfully reflect the actual communication cost (e.g., time) of a protocol.
Moreover, the bandwidth perspective may fail to capture the inherent cost of message passing.
In such a model, a message of one megabyte and a thousand messages of one kilobyte have the identical communication cost, while various communication overheads are ignored.

We adopt the simple linear model of~\cite{fraigniaud1994methods}, in which the cost to pass a message of size~$m$ is~$\alpha+\beta\cdot m$, where~$\alpha$ is the startup time, and~$\beta$ refers to the per-unit data transfer cost.
Although other communication models such as BSP~\cite{valiant1990bridging}, LogP~\cite{culler1993logp} and Postal~\cite{barnoy92postal} are relevant to our setting, we adopt the classic linear cost model for ease of algorithm design and analysis.

Therefore, we assume that the system operates in consecutive communication~\emph{rounds}.
During one round, every port in the network is utilized at most once.
As a result, round~$t$ incurs a time cost of~$\alpha+\beta\cdot m_t$, where~$m_t$ is the size of the largest message (measured in the number of bits) transferred among all ports of all processors in that round. We focus on two communication metrics:
\begin{itemize}
    \item $C_1$: the number of rounds incurred by completing the decentralized encoding.
    \item $C_2$: the number of elements in~$\bbF_q$ transferred in a sequence during the decentralized encoding, i.e.,~$C_2=\sum_{t\in[C_1]}m_t$.
\end{itemize}
Our main goal in this paper is to improve the total communication cost given by
\begin{equation*}
    \textstyle C=\alpha\cdot C_1+\left(\beta\ceil{\log_2q}\right) \cdot C_2.
\end{equation*}

\subsection*{Our Contributions}

In this paper, we investigate the decentralized encoding problem and provide several efficient solutions.
Specifically,
\begin{itemize}
    \item We propose a framework that efficiently addresses the decentralized encoding problem for systematic codes.
    The framework centers around~\emph{all-to-all encode}, a new collective operation concerning a setting in which every processor simultaneously serves both as a source and as a sink. 
    That is, every processor holds a initial data symbol in~$\bbF_q$ and requires a linear combination of all data; the linear combinations are determined by columns of some square matrix (Section~\ref{section:framework}).
    Furthermore, we extend our decentralized encoding framework to non-systematic codes (Appendix~\ref{appendix:non-sys}).

    \item We provide implementations of all-to-all encode operation on different levels.
    On the universal level, we develop a communication-efficient universal algorithm that implements all-to-all encode for all possible square coding matrices, which is optimal in~$C_1$ and achieves the lower bound on~$C_2$ within a constant factor of~$\sqrt{2}$ (Section~\ref{section:universal}).
    On the specific level, we present a family of all-to-all encode algorithms tailored for either Vandermonde, Cauchy-like (see~\eqref{eq:define-cauchy-like} below), or Lagrange matrices (Remark~\ref{remark:lagrange}), which optimize the aforementioned universal algorithm in terms of~$C_2$ (Section~\ref{section:vandermonde}).

    \item We combine aforementioned results and provide a decentralized encoding solution for systematic Reed-Solomon codes with a significant gain in communication costs over applying the universal all-to-all encode algorithm.
    This solution can be extended to Lagrange codes (Section~\ref{section:reedsolo}).
\end{itemize}

\begin{figure}
    \centering
    \includegraphics[width=0.48\textwidth]{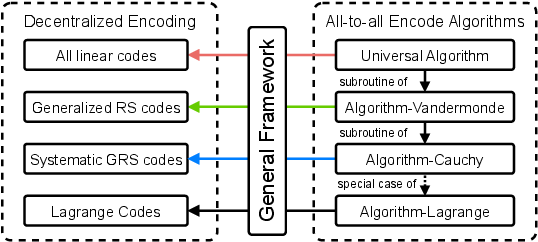}
    \caption{
    Summary of our contributions.
    First, we provide a general framework to the decentralized encoding problem, which rely on a newly defined cooperative communication operation called~\emph{all-to-all encode}.
    We provide a universal algorithm to implement all-to-all encode for all possible coding matrices.
    The universal algorithm serve as a subroutine to implement all-to-all encode specifically for Vandermonde matrices, which can be used for decentralized encoding for generalized Reed-Solomon codes.
    It then serves as a subroutine of algorithms for Cauchy-like and Lagrange matrices for systematic generalized Reed-Solomon codes and Lagrange codes, respectively.
   }
    \label{fig:flow}
\end{figure}

Figure~\ref{fig:flow} provides a graphical illustration of our contribution.

\section{Related works}\label{section:related-works}

    Collective communication operations (e.g., broadcast, reduce, scatter, gather, etc.) are fundamental building blocks for communication patterns in distributed and parallel systems.
    They concern data movement between participants in a network, and have been studied extensively due to their importance in parallel algorithms under different network models, see~\cite{grama2003} for a thorough introduction to the topic. 

    Most closely related to our work, Jeong et al.~\cite{jeong2018masterless} study the decentralized encoding process of systematic MDS codes as a sub-routine of a coded Fast-Fourier Transform algorithm.
    Specifically,~\cite{jeong2018masterless} provides a universal algorithm called~\emph{multi-reduce} based on the~\emph{broadcast} and~\emph{all-gather} operations.
    However, multi-reduce is based on the one-port model, i.e.,~$p=1$, and assumes that~$R\mid K$, while our setting is generalized for any~$R,K$ and~$p$.
    It also incurs~$(R-2\sqrt{R}-1)(\beta\ceil{\log_2q})W$ more communications than our proposed method.
    Furthermore, specific algorithms are not discussed in~\cite{jeong2018masterless}.

In the distributed storage literature, \cite{dimakis2006decentralized} addresses the problem of decentralized encoding of a \textit{random} code. That is, distributed source processors pass on a random linear combination of packets to their neighboring parity nodes, resulting in an MDS code with high probability.
However, \cite{dimakis2006decentralized} did not study the decentralized encoding problem any further.

\section{General framework}\label{section:framework}
 
The presented solution to the decentralized encoding problem is based on three collective operations, two existing ones and a recent one defined in~\cite{wang2022all}.
They serve as basic building blocks and their definitions are given as follows.

\begin{definition}\textbf{(One-to-all broadcast)}~\cite[Ch.~13.1]{sanders2019sequential}
Consider a distributed system with~$K$ processors.
Initially, the sole source processor~$\rmP_0$ possesses some data packet~$x_0$.
At the termination of the operation, every processor~$\mathrm{P}_k$ obtains~$x_0$, for~$k\in[0,K-1]\triangleq\{ 0,1,\ldots,K-1\}$.
\end{definition}

\begin{definition}\textbf{(All-to-one reduce)}~\cite[Ch.~13.2]{sanders2019sequential}\label{definition:all2onereduce}
Consider a distributed system with~$K$ processors.
Initially, every processor~$\rmP_k$ possesses a data packet~$x_k$, for~$k\in[0,K-1]\triangleq\{ 0,1,\ldots,K-1\}$.
At the termination of the operation, the sole sink processor~$\rmP_0$ obtains~$x_0\oplus x_1\oplus\cdots\oplus x_{K-1}$, i.e., the accumulation of all data defined by some commutative and associative operator~$\oplus$.
\end{definition}

As the name suggests, the one-to-all broadcast operation employs a single source processor to share its data with every other processor in the system.
Meanwhile, the all-to-one reduce operation allows a sink processor to learn the reduction, e.g., the sum of the data from all processors.
Note that one of them is the dual of the other, i.e., a procedure that completes broadcast can be utilized to complete reduce by reversing the order of communication.

\begin{figure}
    \centering
    \scalebox{.52}{\input{figures/example-universal}}
    \caption{An example of completing all-to-all encode for any~$\bfC\in\bbF_q^{4\times 4}$ in~$2$ rounds in a system with $p=1$ ports, in which indices of processors are computed modulo~$4$.
    Every processor~$k\in[0,3]$ requires~$\coded{x}_k=\bfC_{0k}x_0+\bfC_{1k}x_1+\bfC_{2k}x_2+\bfC_{3k}\bfx_3$.
    In the first round, it receives~$x_{k-1}$ from processor~$\rmP_{k-1}$.
    In the second round, it receives~$\bfC_{(k-2)k}x_{k-2}+\bfC_{(k-3)k}x_{k-3}$ from processor~$\rmP_{k-2}$. 
    Finally, processor~$\rmP_k$ combines the~$2$ received data with~$\bfC_{kk}x_k$ and obtains~$\coded{x}_k$.}\label{figure:universal}
\end{figure}

\begin{definition}
\textbf{(All-to-all encode)}~\cite[Definition 1]{wang2022all}
Consider a distributed system with~$K$ processors.
Every processor~$\rmP_k$ for~$k\in[0,K-1]$ possesses~\emph{initial data}~$x_k\in\bbF_q$, and obtains a coded data~$\coded{x}_k\in\bbF_q$ at the termination of the operation.
For a matrix~$\bfC\in\bbF_q^{K\times K}$ that is known a priori to all processors, the coded data satisfies
\begin{equation*}
 \left(\coded{x}_0,\ldots,\coded{x}_{K-1}\right) = ({x_0},\ldots,{x_{K-1}})\cdot \bfC.
\end{equation*}
\end{definition}

That is, all-to-all encode does not differentiate between sources and sinks, as every processor simultaneously serves both roles.
In particular, the all-to-all encode operation concerns a scenario in which every processor has its own data treated as an element of some finite field~$\bbF_q$, and requires a linear combination of all other vectors in the system.
An algorithm which successfully completes all-to-all encode for a given~$\bfC$ is said to \emph{compute~$\bfC$} (see Figure~\ref{figure:universal} for an illustrative example, and note that the operation is not limited to the case where~$K=4$ and~$p=1$).

\begin{table*}[]
\centering
\resizebox{\textwidth}{!}{%
\begin{tabular}{|c|cl|c|}
\hline
                                          & \multicolumn{2}{c|}{Communication Cost}                                                                          & Note    \\ \hline
\textit{Universal Algorithm.}             & \multicolumn{2}{c|}{$C_\text{A2A,Univ}(K)=\alpha\cdot \ceil{\log_{p+1}K}+ \beta \log_2q \cdot{2\sqrt{K}}/{p}$}   & Theorem~\ref{theorem:universal-c} \\ \hline
\textit{Specific Algorithm (DFT)}         & \multicolumn{2}{c|}{$C_\text{A2A,DFT}(K)=H\cdot C_\text{A2A,Univ}(P)$, where $K=P^H$}                            & Theorem~\ref{theorem:dft-cost} \\ \hline
\textit{Specific Algorithm (Vandermonde)} & \multicolumn{2}{c|}{$C_\text{A2A,Vand}(K)= C_\text{A2A,DFT}(P)+C_\text{A2A,Univ}(M)$, where $K=M\cdot P^H$} & Theorem~\ref{theorem:Vandermonde} \\ \hline
\end{tabular}%
}
\caption{Costs of the all-to-all schemes.}
\label{tab:costs}
\end{table*}

Unlike broadcast and reduce, all-to-all encode is a newly defined collective operation that has not yet been studied.
An efficient algorithm that implements all-to-all encode for every possible matrix is provided in Section~\ref{section:universal}, which is then optimized for specific matrices.
Their communication costs can be found in Table~\ref{tab:costs}.

\begin{remark}
    The all-to-all encode operation considers that every processor holds an element in some finite field~$\bbF_q$, instead of a vector of such elements.
    Nevertheless, an all-to-all encode algorithm can be applied in the latter case by regarding data vectors in~$\bbF_q^W$ as an element in the extension field~$\bbF_{q^W}$, while keeping the coding matrix~$\bfC$ in the base field~$\bbF_q$.
    This simple extension does not introduce extra rounds ($C_1)$, but incurs~$W$ times more communications in terms of~$C_2$.
    For a thorough introduction on the relationship between the base field~$\bbF_q$ and its extension field~$\bbF_{q^W}$, see~\cite[Chapter 3]{roth2006coding}.
\end{remark}

With these three building blocks defined, we proceed to introduce the general framework of the presented solution to the decentralized encoding problem.
An extension to non-systematic codes is given in Appendix~\ref{appendix:non-sys}.

\begin{remark}
    This section focuses on a framework for decentralized encoding, and considers the three operations as modular components, whose implementations may vary to suit different needs.
    However, the implementations of these operations significantly affect the communication costs in solving the decentralized encoding problem.
    
    The broadcast and the reduce operations have been well-studied in the literature (recall that one is the dual of the other and hence the implementation of broadcast can be used in reduce).
    A detailed discussion on the implementation details of these two operations is given in Appendix~\ref{app:broadcast}.

    In later sections, we provide implementations of the all-to-all encode operation for different matrices.
    In Section~\ref{section:universal}, we introduce an optimal universal algorithm that completes it for any square matrix~$\bfC$, so that we have a procedure for decentralized encoding for any systematic linear code.
    In Section~\ref{section:vandermonde}, we provide an algorithm that implements all-to-all encode specifically for Vandermonde matrices which outperforms the universal one in communication costs; it is then employed in Section~\ref{section:reedsolo} to compute Cauchy-like matrices, providing a decentralized encoding algorithm tailored on systematic Reed-Solomon codes, and to compute Lagrange matrices, providing a decentralized algorithm for Lagrange codes.
\end{remark}

\subsection{The case of \texorpdfstring{$K\geq R$}{K>=R}}

We first consider the case where~$K=R\cdot M$ for some integer~$M$.
The matrix~$\bfA\in\bbF_q^{K\times R}$ is treated as a stack of square matrices~$\bfA_0,\ldots,\bfA_{M-1}\in\bbF_q^{R\times R}$, i.e.,
\small
\begin{equation}\label{eq:submatrix-of-A}
\bfG=\left[\begin{array}{ccc|ccc} 
  &   &    &  &    &      \\
  & \mbox{\normalfont\LARGE\bfseries I} &    &    &   \mbox{\normalfont\LARGE\bfseries A}   \\
  &   &    &  &    &      \\
\end{array}\right]=
\left[\begin{array}{ccc|c}
  &   &    &  \bfA_0  \\
  & \mbox{\normalfont\LARGE\bfseries I} &    & \vdots \\
  &   &    & \bfA_{M-1} \\
\end{array}\right].
\end{equation}
\normalsize

For convenience of indexing, we organize the~$K$ sources in an~$R\times M$ grid.
For index~$k=r+m\cdot R$ in which~$r\in[0,R-1]$ and~$m\in[0,M-1]$, processor~$\rmS_k$ is renamed to~$\rmS_{r,m}$ and placed in row~$r$ and column~$m$.
Based on this grid, the proposed framework involves two phases, one with column-wise operations followed by another with row-wise operations, described as follows.

In phase one, we initiate~$M$ instances of column-wise all-to-all encode operating in parallel, i.e., for~$m\in[0,M-1]$, processors in column~$m$ perform an all-to-all encode algorithm to compute the matrix~$\bfA_m$.
By the definition of all-to-all encode, processor~$\rmS_{r,m}$ obtains
\begin{equation*}
    \coded{\bfx}_{r,m}^\intercal\triangleq(\bfx_{m\cdot R}^\intercal,\bfx_{m\cdot R+1}^\intercal,\ldots,\bfx_{(m+1)\cdot R-1}^\intercal)\cdot (\bfA_m)_r,    
\end{equation*}
where~$(\bfA_m)_r$ is the~$r$-th column of~$\bfA_m$.
Clearly, the coded packet~$\coded{\bfx}_r=\sum_{m\in[0,M-1]}\coded{\bfx}_{r,m}$ required by sink processor~$\rmT_r$ is the sum of
\begin{equation*}
    \coded{\bfx}_{r,0},\coded{\bfx}_{r,1},\ldots,\coded{\bfx}_{r,M-1},
\end{equation*}
i.e., the (partially) coded data held by all processors in row~$r$.
As a result, phase two of the framework involves~$R$ instances of row-wise all-to-one reduce, and after each one, the partially coded packets are accumulated in the sink.

\begin{figure}
    \centering
\scalebox{0.62}{\input{figures/example-framework-1}}
    \caption{An example of completing decentralized encoding for matrix~$\bfA$ when~$p=1$,~$K=25$ and~$R=4$.
    In phase one, source processors are placed in a~$4\times7$ grid, and sink processors~$\rmT_1,\rmT_2$ and~$\rmT_3$ are borrowed to complete its last column, while acting as source processors which hold the zero packet.
    Processors in each  column perform all-to-all encode for a square submatrix of~$\bfA$, and obtain (partially) coded packets.
    In phase two, processors in each row perform all-to-one reduce, and accumulate the partial packets to the rightmost processor.}
    \label{fig:framework-1}
\end{figure}

In the case where~$R\nmid K$, we let~$M=\ceil{K/R}, L=K\bmod R$, and~$\bfB$ be an~\emph{arbitrary} matrix in~$\bbF_q^{L\times K}$.
We create~$\bfA'\in\bbF_q^{RM\times R}$ by placing~$\bfB$ at the bottom of~$\bfA$, and treat~$\bfA'$ as a stack of~$\bfA_0,\ldots,\bfA_{M-1}\in\bbF^{R\times R}$.
Moreover, we ``borrow''~$L$ sink processors to complete the last column of the grid, and each of them holds an empty ($\mathbf{0}$ in~$\bbF_q^W$) vector.
As in the case of~$R|K$, processors in each column perform all-to-all encode and obtain partially coded data.
In the second phase, the partially coded data are accumulated in the rightmost processors (sinks) after row-wise all-to-one reduce operations, and hence completes decentralized encoding.
Note that the choice of matrix~$\bfB$ does~\emph{not} affect the result since the ``borrowed'' processors hold empty data.
For an illustrative example, see Figure~\ref{fig:framework-1}.

In order to combine the two cases $R\vert K$ and~$R\nmid K$, let~$M=\ceil{K/R}$.
The communication cost of phase one depends on the underlying procedure that implements all-to-all encode, and for phase two, it depends on the algorithm that implements reduce.
This yields the following theorem, in which~$C_\text{A2A}(\bfA_m)$ measures the communication cost of all-to-all encode for square matrix~$\bfA_m$, and $C_\text{BR}(N,W)$ measures the communication cost of broadcasting a vector of length~$W$ in a system of~$N$ processors.
\begin{theorem}
    When~$K\geq R$, the proposed framework completes decentralized encoding with communication cost
    \begin{equation*}
        C= \max\left[ {C_\text{A2A}(\bfA_0),\ldots,C_\text{A2A}(\bfA_{M-1})}\right]+C_\text{BR}\left( \ceil{\textstyle\frac{K}{R}},W\right).
    \end{equation*}
\end{theorem}

\subsection{The case of \texorpdfstring{$K< R$}{K<R}}

We first consider the case where~$R=K\cdot M$ for some integer~$M$.
The matrix~$\bfA\in\bbF_q^{K\times R}$ is treated as a concatenation of square matrices~$\bfA_0,\ldots,\bfA_{M-1}\in\bbF_q^{K\times K}$, i.e.,
\begin{equation}\label{eq:submatrix-of-A-horizontal}
\bfG=\left[\begin{array}{c|c} \bfI & \bfA\end{array}\right]
=\left[\begin{array}{c|ccc}  \bfI &  \bfA_0 & \cdots & \bfA_{M-1}\\\end{array}\right].
\end{equation}
 
We organize the~$R$ sinks in an~$K\times M$ grid in which for~$r=k+m\cdot K$, processor~$\rmT_r$, now renamed as $\rmT_{k,m}$, is placed in row~$k\in[0,K-1]$ and column~$m\in[0,M-1]$.
For ease of demonstration, the source processors are appended as column~$M$ of the processor grid. 

In phase one, we initiate~$K$ instances of one-to-all broadcast in parallel, which allows each source to share its initial data with other processors in the same row, i.e., for every~$k\in[0,K-1]$,~$\bfx_k$ is disseminated to processor~$\rmT_{k,m}$ for every~$m\in[0,M-1]$.
In phase two, sink processors in column~$m$ perform all-to-all encode for computing the matrix~$\bfA_m$.
After that, every sink processor~$\rmS_{k,m}$ obtains~$\coded{\bfx}_k$, i.e., the linear combination of~$\bfx_0,\ldots,\bfx_{K-1}$ defined by column~$(k+K\cdot m)$ of~$\bfA$, which is also the column~$k$ of~$\bfA_m$.

\begin{figure}
    \centering
    \scalebox{0.62}{\input{figures/example-framework-2}}
    \caption{An example of completing decentralized encoding for matrix~$\bfA$ when~$p=1$,~$K=4$ and~$R=25$.
    In phase one, sink processors are placed in a~$4\times7$ grid, and sources (as placed in the rightmost column) broadcast their initial data to others in the same row. 
    In phase two, source processors~$\rmS_1,\rmS_2$ and~$\rmS_3$ are borrowed to complete column~$6$ of the grid, and hence every processor in column~$m\in[0,6]$ holds~$\bfx_m$.
    Next, processors in each column perform all-to-all encode for a square submatrix of~$\bfA$ using the data received in the previous phase, and complete decentralized encoding.}
    \label{fig:framework-2}
\end{figure}

In the case where~$K\nmid R$, we let~$M=\ceil{R/K}$ and~$L=R\mod K$.
As in the case of~$K\vert R$, phase one involves~$K$ row-wise one-to-all broadcasts.
However, the processor grid has~$L$ (unfilled) rows of~$M-1$ processors, and they may finish earlier than others.
Nevertheless, we enable them to wait for the rest processors to finish this phase.

In phase two, we~``borrow'' source~$\rmS_k$ in order to complete row~$k$, if it is unfilled and contains only~$M-1$ processors.
This makes a full~$K\times M$ grid in which every processor in row~$k$ has~$\bfx_k$.
Moreover, we append an~\emph{arbitrary} matrix~$\bfB\in\bbF_q^{K\cdot L}$ to the right of~$\bfA$.
This results in~$\bfG'=[\bfI\mid \bfA']$ for~$\bfA'\in\bbF_q^{K\times R}$, and we refer to~$\bfA'$ as a concatenation of~$\bfA_0,\ldots,\bfA_{M-1}\in\bbF_q^{R\times R}$.
Finally, every column~$m$ performs all-to-all encode for matrix~$\bfA_m$, and concludes decentralized encoding.
Note that the choice of~$\bfB$ does not affect the decentralized encoding result, as it determines the coded packets required by the ``borrowed'' processors, which are essentially source processors that do not require any encoded packet.
For an illustrative example, see Figure~\ref{fig:framework-2}.

In order to combine the two cases $R\vert K$ and~$R\nmid K$, let~$M=\ceil{R/K}$.
Phase one is solely broadcast, and for phase two, the communication cost depends on the all-to-all encode algorithm.
This yields the following theorem.
\begin{theorem}
    When~$K<R$, the proposed framework completes decentralized encoding with communication cost
    \begin{equation*}
        C= \max\left[ {C_\text{A2A}(\bfA_0),\ldots,C_\text{A2A}(\bfA_{M-1})}\right]
        +C_\text{BR}\left( \ceil{\textstyle\frac{R}{K}},W\right).
    \end{equation*}
\end{theorem}

\begin{remark}
    For ease of demonstration, we assume that~$K\vert R$ or~$R\vert K$ in the rest of this paper.
\end{remark}
 
\section{Universal all-to-all encode algorithm}\label{section:universal}
Recall that an algorithm needs to determine which processor communicates with which other processors at each round (scheduling), and how the sent messages are obtained as linear combinations of previously received data (coding scheme).
A universal algorithm is one which completes an all-to-all encode operation for any square matrix~$\bfC\in\bbF_q^{K\times K}$ by only varying the coding scheme, with the scheduling fixed.

Universal algorithms are generally preferred in the scenarios where communications between processors need to be scheduled prior to knowing~$\bfC$, or are utilized in consecutive computations of different coding matrices.
In addition, in the case where the~$\bfA_0,\ldots,\bfA_{M-1}$ as defined in~\eqref{eq:submatrix-of-A} or~\eqref{eq:submatrix-of-A-horizontal} do not share a common structure (e.g,. Vandermonde or Cauchy), a universal algorithm may be chosen for the simplicity or uniformity of system design.

\subsection{Lower bounds}
We first propose the lower bounds for~$C_1$ and~$C_2$, which apply to any universal algorithm.
\begin{lemma}\label{lemma:universalC1bound}
Any universal algorithm has~$C_1\geq \ceil{\log_{p+1}K}$.
\end{lemma}
\begin{proof}
Since a universal algorithm must apply for all~$K\times K$ matrices, it must also apply for matrices with no zero entries, i.e., where every processor requires a linear combination of all initial data.
Similar to~\cite[Proposition 2.1]{bruck1997efficient}, a data packet~$x_k$ can reach at most~$(p+1)^t$ processors in the~$p$-port model after round~$t$. 
Hence, it takes at least~$\ceil{\log_{p+1}K}$ rounds for any data~$x_k$ to reach all processors.
\end{proof}

\begin{lemma}\label{lemma:universalC2bound}
Any universal algorithm has~$C_2\geq \frac{\sqrt{2K}}{p}-O(1)$.
\end{lemma}
\begin{proof}
We define~\emph{baseline algorithms} as a class of universal algorithms in which every processor passes exactly one field element in~$\bbF_q$ through each port during every round; clearly, all baseline algorithms have~$C_1=C_2$.

Given a universal algorithm in which during round~$t$, processors transfer at most~$d_t$ field elements through one of the ports (and thus~$C_2=\sum_t d_t$), there exists a corresponding baseline algorithm that ``simulates'' each round~$t$ in the universal algorithm using~$d_t$ rounds, and thus has the same~$C_2=\sum_t d_t$.
Therefore, a universal algorithm cannot outperform all baseline algorithms in terms of~$C_2$.
Hence, it suffices to bound~$C_1$, and hence~$C_2$, for baseline algorithms.

We provide a lower bound on~$C_1$ for baseline algorithms by counting possible~\emph{coding schemes}, i.e., the number of ways processors can linearly combine previously received packets during the algorithm. 
Notice that once the scheduling and coding scheme are fixed, an algorithm cannot compute two distinct matrices\footnote{Recall that an algorithm computes a matrix~$\bfC$ only if it computes the product~$(x_0,\ldots,x_{K-1})\cdot\bfC$ for every~$(x_0,\ldots,x_{K-1})\in\bbF_q^{K}$. 
Hence, if~$(x_0,\ldots,x_{K-1})\cdot\bfC=(x_0,\ldots,x_{K-1})\cdot\bfD$ for every~$(x_0,\ldots,x_{K-1})\in\bbF_q^K$, then~$\ker(\bfC-\bfD)=\bbF_q^{K}$, which implies~$\bfC=\bfD$.}. 
Therefore, this value must be greater than or equal to~$|\bbF_q^{K\times K}|=q^{K^2}$.

Assume a baseline algorithm with~$C_1=T$, and let~$\bfs_k^{(t)}$ be the vector of packets received by processor~$\rmP_k$ prior to the beginning of round~$t\in[T]$.
Clearly, we have~$|\bfs_k^{(t)}|=(t-1)p$.
In round~$t$, for every port~$\rho\in[0,p-1]$, the processor creates and sends a packet~$y_{k,\rho}^{(t)}$ by summing the initial packet~$\bfx_k$ with a linear combination of previously received packets\footnote{Since the receiver will linearly combine the received packets in future rounds of network coding, the initial packet~$x_k$ has coefficient~$1$ to avoid overcounting the coding schemes.}, i.e., 
$$
y_{k,p}^{(t)} = x_k+ (\bfg_k^{(t)})^\intercal \cdot\bfs_k^{(t)},
$$
where~$\bfg_k^{(t)}\in\bbF_q^{p(t-1)}$ is a \emph{coding vector} defining the linear combination of elements in~$\bfs_k^{(t)}$. 
The number of possible coding vectors in round~$t$ is then~$q^{p(t-1)}$. 

Note that there are~$K$ processors in the~$p$-port system, and the baseline algorithm contains~$T$ rounds.
Further, at the end of the~$T$-round baseline algorithm, processor~$\rmP_k$ obtains~$\coded{\bfx}_k$ by linearly combining the~$Tp+1$ received packets, including the initial~$\bfx_k$. 
Thus, the total number of coding schemes is
\begin{equation}\label{eq:inequality}
   \textstyle  q^{K(Tp+1)}\cdot \prod_{t\in[T]}q^{(t-1)Kp^2}\geq q^{K\cdot K}.
\end{equation}
Taking logarithms on both sides of~\eqref{eq:inequality}, we have
\begin{equation*}
   \textstyle K(Tp+1)+\sum_{t\in[T]} (t-1)Kp^2 \geq K^2,
\end{equation*}
which simplifies to~$p^2T^2-p(p-2)T + 2 (1-K) \geq 0$.
Hence,
\begin{equation*}
   \textstyle C_1=T\geq\frac{1}{2}-\frac{1}{p}+\sqrt{\frac{1}{4}-\frac{1}{p}-\frac{1}{p^2}+\frac{2K}{p^2}} =\frac{\sqrt{2K}}{p}-O(1).\qedhere
\end{equation*}
\end{proof}
\begin{remark}[Lower bound for specific algorithms]\label{remark:bounds4SpecificMatrices}
Clearly, specific algorithms can perform at least as well as universal ones in any figure of merit, since any universal algorithm subsumes specific algorithms for all matrices by definition. 
It is readily verified that any matrix which contains a non-zero row cannot be computed with~$C_1 < \log_{p+1}K$ and~$C_2 < \log_{p+1}K$; this is due to the fact a data packet cannot be disseminated to all~$K$ processors in less than this many rounds. 
\end{remark}

\subsection{Prepare-and-shoot: An optimal universal algorithm}

\begin{figure}[t]
    \centering
    \includegraphics[width=0.48\textwidth]{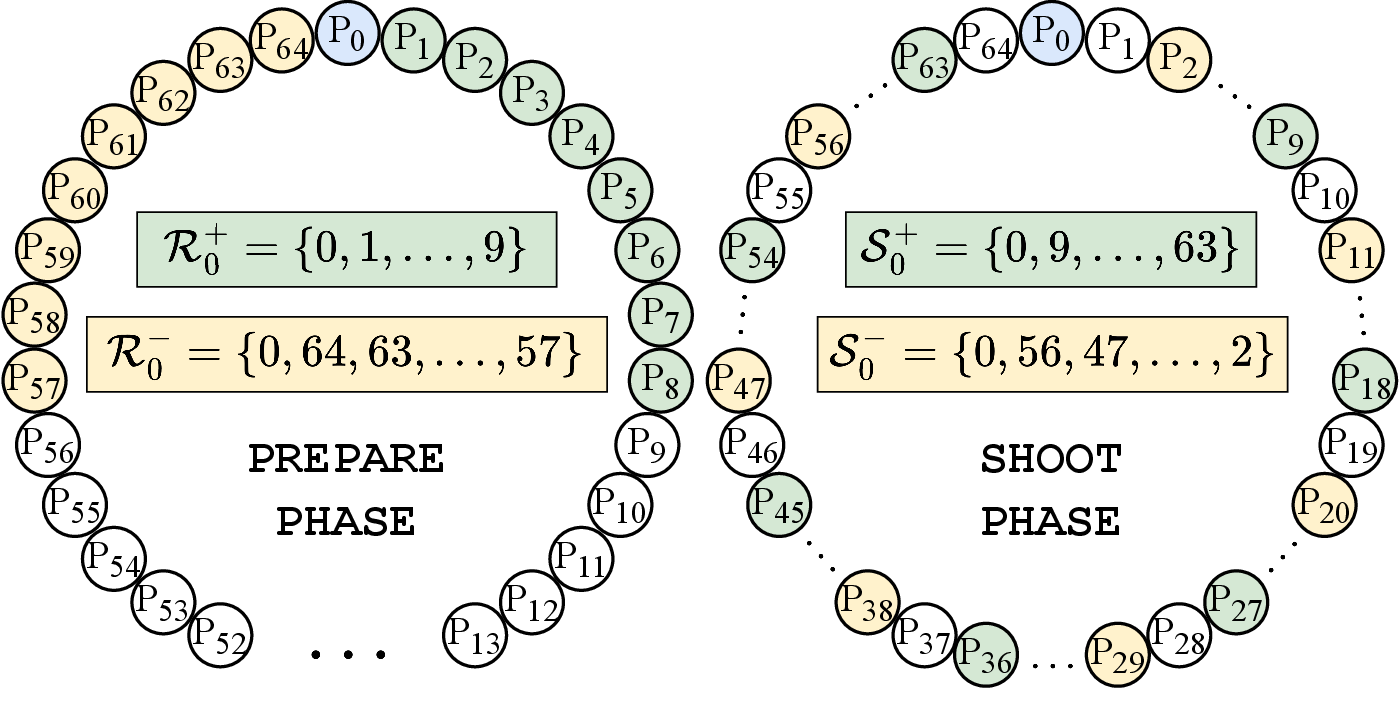}
    \caption{
    Overview of an all-to-all encode process with~$K=65$ and~$p=2$ from~$\rmP_0$'s perspective.
    During the prepare phase, processor~$\mathrm{P}_0$ obtains the initial packets from processors in~$\cR_0^-$, and disseminates its own initial packets to processors in~$\cR_0^+$.
    In the shoot phase,~$\rmP$ computes and disseminates partially coded packets to processors in~$\cS_0^+$ using the packets received, and obtains its coded packet by summing the partially coded packets received from processors in~$\cS_0^-$.}
    \label{fig:prepare-and-shoot}
\end{figure}

Inspired by the all-gather algorithm developed by Bruck \textit{et al.}~\cite{bruck1997efficient}, we present a universal algorithm consisting of two phases called~\emph{prepare} and~\emph{shoot}.

Let~$L=\ceil{\log_{p+1}K}$, which equals the total number of rounds required to complete the proposed algorithm.
If~$L$ is odd, then the prepare phase requires~$T_p=(L+1)/2$ rounds, and the shoot phase requires~$T_s=(L-1)/2$ rounds.
Otherwise,~$T_p=T_s=L/2$.

To describe the phases, let~$m=(p+1)^{T_p},n=\ceil{K/m}$, and hence~$(n-1)m<K\leq nm$. 
For~$k\in[0,K-1]$, let
\small
\begin{align*}
    \cR_k^+&=\{k+\ell\vert \ell\in[0,m-1]\},
    \cS^+_k=\{k+\ell\cdot m\vert \ell\in[0,n-1]\},\\
    \cR_k^-&=\{k-\ell\vert \ell\in[0,m-1]\},
    \cS^-_k=\{k-\ell\cdot m\vert \ell\in[0,n-1]\}.
\end{align*}
\normalsize
In a nutshell,~$\cR_k^+$ and~$\cR_k^-$ contain the indices of processors that processor~$\rmP_k$ communicates with in the prepare phase, while~$\cS^+_k$ and~$\cS^-_k$ contain the indices of processors that processor~$\rmP_k$ communicates with in the shoot phase, directly or indirectly.
For convenience of notation, indices of processors are computed~$\bmod K$.
An overview of the proposed algorithm with examples of these sets can be found in Figure~\ref{fig:prepare-and-shoot}.

\begin{figure}[t]
    \centering
    \includegraphics[width=0.48\textwidth]{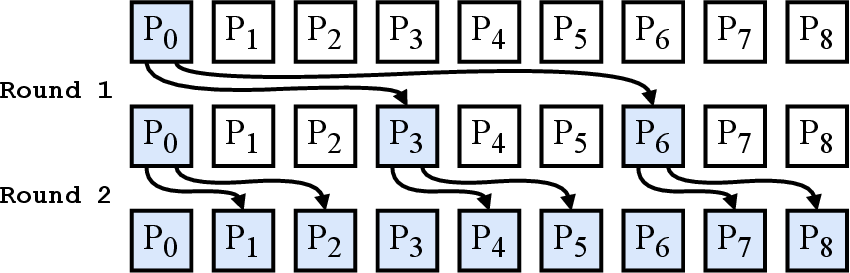}
    \caption{
    Illustration of the dissemination of~$\bfx_0$ in the prepare phase.
    In round one, processor~$0$ broadcasts~$\bfx_0$ to processor~$3$ and processor~$6$.
    In round two, processor~$0/3/6$ passes~$x_0$ to processor~$1/4/7$ and processor~$2/5/8$, and concludes the prepare phase.
    Note that this figure focuses on an individual broadcast operation, and there are~$K$ broadcasts happening in parallel.
    }
    \label{fig:prepare}
\end{figure}

\begin{algorithm}
\caption{Prepare Phase (for processor~$\rmP_k$)}\label{alg:pp}
\begin{algorithmic}[1]
    \State Let~$\texttt{memory}$ be an empty set of packets.
    \State $\texttt{memory}_k\gets\{x_k\}$\label{line:initial-mem}
    \For{$t\gets 1,2,\ldots,T_p$}
        \For{$\rho=1$ to~$p$}\Comment{As a sender}
            \State send~$\texttt{memory}_k$ to~$\rmP_{k+\rho\cdot (p+1)^{T_p-t}}$
        \EndFor
        \For{$\rho=1$ to~$p$}\Comment{As a receiver}
            \State \small receive~$\texttt{memory}_{k-\rho\cdot (p+1)^{T_p-t}}$ from~$\rmP_{k-\rho\cdot (p+1)^{T_p-t}}$\normalsize
            \State $\texttt{memory}_k\gets\texttt{memory}_k\cup\texttt{memory}_{k-\rho\cdot (p+1)^{T_p-t}}$
        \EndFor
    \EndFor
    \State \textbf{Output}~$\texttt{memory}$
\end{algorithmic}
\end{algorithm}

\noindent\textbf{Prepare phase:}
This phase consists of~$K$ one-to-$m$ broadcast operations, executed in parallel using the folklore~$(p+1)$-nomial tree structure.
Each individual broadcast can be interpreted recursively:
In the first round, the broadcaster sends its packet to~$p$ other processors, resulting in~$p+1$ processors possessing the packet. Consequently, after each round, the problem is divided into~$p+1$ smaller subproblems, which can be solved recursively in~$\ceil{\log_{p+1}K}$ rounds.


Let~$\texttt{memory}_k$ be the internal storage of processor~$\rmP_k$.
Initially,~$\texttt{memory}_k=\{x_k\}$ (line~\ref{line:initial-mem}).
In round~$t\in[1,T_p]$, for each~$k\in[0,K-1]$, processor~$\rmP_k$ forwards the entire content of its \texttt{memory}, which is a set consisting of all initial packets it possesses~$\texttt{memory}_k$ to
\begin{equation*}
    \rmP_{k+(p+1)^{T_p-t}},    \rmP_{k+2(p+1)^{T_p-t}},\ldots,    \rmP_{k+p(p+1)^{T_p-t}}.
\end{equation*}
Simultaneously, it receives messages from
\begin{equation*}
    \rmP_{k-(p+1)^{T_p-t}},    \rmP_{k-2(p+1)^{T_p-t}},\ldots,    \rmP_{k-p(p+1)^{T_p-t}},
\end{equation*}
and merges the contents with its own memory.
See Algorithm~\ref{alg:pp} for details and Figure~\ref{fig:prepare} for an illustrative example.

\begin{lemma}\label{lemma:prepare}
After the prepare phase, every processor~$\rmP_k$ has obtained~$\bfx_r$ for every~$r\in\cR_k^-$, with
\begin{equation*}
    \textstyle C_{1,\text{prep}}=T_p,C_{2,\text{prep}}= \frac{(p+1)^{T_p}-1}{p}.
\end{equation*}
\end{lemma}
\begin{proof}
Consider the following tree defined recursively, which contains every processor that has obtained~$x_k$: Initially,~$\rmP_k$ is added as the root. 
In every subsequent round~$t$, every processor~$\rmP_r$ in the current tree is connected to the processor with index~$r+\rho\cdot(p+1)^{T_p-t}$, for every~$\rho\in[0,p-1]$ (i.e., the processors receiving messages from processor~$\rmP_r$ in round~$t$). 
Obviously, the recursion is concluded after~$T_p$ rounds, and the tree contains all processors in~$\cR_k^+$.
That is, every processor in~$\cR_k^+$ obtains~$\bfx_k$ after~$T_p$ rounds.
Therefore, processor~$\rmP_k$ obtains~$\bfx_r$ for every~$r\in\cR^-_k$ after the shoot phase.

Further, as the number of field elements in a message increases by~$p$-fold after each round, we have
\begin{equation*}
  C_{2,\text{prep}}=\sum_{t=1}^{T_p}(p+1)^{t-1}=\frac{(p+1)^{T_p}-1}{p}.   \qedhere
\end{equation*}
\end{proof}

\begin{figure}[t]
    \centering
    \includegraphics[width=0.48\textwidth]{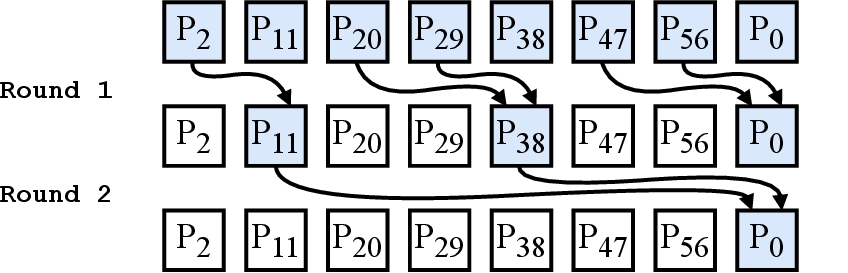}
    \caption{Illustration of the reduce of~$n=8$ packets, i.e.,~$w_{2,0},\ldots,w_{56,0},w_{0,0}$, to processor~$\rmP_k$. In every round, processors sum up the received packets and forward the result, and after the operation, processor~$\rmP_0$ obtains~$y_0=\sum_{s\in\cS_0^-}w_{s,0}$. Note that~$w_{2,0}=(\mathbf{0},x_{59},\ldots,x_{2},\mathbf{0}) \cdot \bfC_{0}$, and~$w_{0,0}=(\mathbf{0},x_{57},\ldots,x_{0},\mathbf{0}) \cdot \bfC_{0}$, processors~$\rmP_0$ computes the coded packet~$\coded{x}_0=y_0-(\mathbf{0},x_{59},\ldots,x_{0},\mathbf{0}) \cdot \bfC_{0}$.
    Note that this figure focuses on one individual reduce operation, and there are~$K$ reduces happening in parallel.
    }
    \label{fig:reduce}
\end{figure}

\noindent\textbf{Shoot phase:} This phase consists of~$K$ $n$-to-one reduce operations using the folklore~$(p+1)$-nomial tree structure happening in parallel.
Each instance of the reduce operation is intended to communicate the correct linear combination of packets to every processor.
As suggested in Section~\ref{section:framework}, the broadcast algorithm in the prepare phase can be utilized to complete reduce operations by reversing the order of communication in each round.

At the beginning of the phase, every~$\rmP_k$ defines~$n$ packets
$$
w_{k,k},w_{k,k+m},\ldots, w_{k,k+(n-1)m}
$$
using the information received in the prepare phase (i.e.,~$x_{r}$ for~$r\in\cR_k^-$), and the coefficients of the matrix~$\bfC$. 
Specifically, for every~$\ell\in[0,n-1]$, processor~$\rmP_k$ initializes
\small\begin{equation*}
    w_{k,k+\ell\cdot m}= (0,\ldots,0,x_{k-m+1},\ldots,x_{k},0,\ldots,0) \cdot \bfC_{k+\ell\cdot m},
\end{equation*}\normalsize
where~$\bfC_{k+\ell\cdot m}$ is the~$(k+\ell\cdot m)$-th column of~$\bfC$.
That is, for~$s\in\cS_k^+$ and~$r\in\cR_k^-$, the packet~$w_{k,s}$ contains a linear combination of $x_r$'s (initial packets that have been obtained by processor~$\rmP_k$) defined by the~$r$-th column of the matrix~$\bfC$.
At the end of this phase, every~$\rmP_k$ obtains    ~$y_k\triangleq\sum_{s\in\cS_k^-}w_{s,k}$.



\begin{remark}
    The packet~$\bfw_{k,s}$ can be viewed as a partially coded packet for~$\rmP_s$, and the objective of the shoot phase is to ensure that each~$\rmP_s$ obtains the sum of all~$\bfw_{k,s}$ with~$k \in \cS_s^-$.
    In other words, the processors in~$\{\rmP_k\}_{k \in \cS_s^-}$ collaboratively assist~$\rmP_s$ in recovering~$y_s$.

    In the special case where~$K = mn$, each set~$\cR_k^-$ consists of distinct indices, i.e., $\cR_k^- \cap \cR_{k'}^- = \emptyset$ for~$k \neq k'$ and~$k,k' \in [K]$, which directly leads to~$\coded{x}_s = y_s$.
    However, in the more general case where~$K < mn$, the sets~$\cR_k^-$ and~$\cR_{k - (n-1)m}^-$ share overlapping indices.
    We refer readers to~\eqref{eq:xkFromyk} for the procedure to compute~$\coded{x}_s$ from~$y_s$ under this setting.
\end{remark}


In round~$t\in[1,T_s]$ of the shoot phase, processor~$\rmP_k$ forwards a message to the processor indexed by~$k+\rho m^t$, and receives a message from the processor indexed by~$k-\rho m^t$, through every port~$\rho\in[p]$. 
To describe this process, we now define a series of trees for every processor in each round, and an illustrative example can be found in Figure~\ref{fig:reduce}.

For every round~$t$, let~$\cT_k^{(t)}$ be a tree defined recursively in~$T_s-t$ steps, as follows.
Initially, processor~$\rmP_k$ is added as the root. 
In every subsequent step~$\tau$, for~$\tau\in[1,T_s-t]$, every processor~$\rmP_r$ in the current tree is connected to processor indexed by~$r+\rho m^{t+\tau}$, for every~$\rho\in[p]$. 
In particular, the tree~$\cT_k^{(T_s)}$ contains only the root~$k$. 
Intuitively, the tree~$\cT_k^{(t)}$ contains processors that the processor~$\rmP_k$ is connecting to, directly or indirectly, in rounds following round~$t$. 

\begin{algorithm}[t]
\caption{Shoot Phase (for processor~$\rmP_k$)}\label{alg:sp}
\begin{algorithmic}[1]
    \State Initialize~$w_{k,k},w_{k,k+m},\ldots,w_{k,k+(n-1)m}$.
    \For{$t\gets 1,2,\ldots,T_s$}
        \For{$\rho=1$ to~$p$}\Comment{As a sender}
            \State send~$w_{k,r}$ to~$\rmP_{k+\rho m^t}$ for every~$r\in \cT^{(t)}_{k+\rho m^t}$
        \EndFor
        \For{$\rho=1$ to~$p$}\Comment{As a receiver}
            \For{$r\in \cT^{(t)}_k$}
                \State receive~$w_{{k-\rho m^t},r}$ from processor~$\rmP_{k-\rho m^t}$
                \State assign~$w_{k,r}\gets w_{k,r}+w_{{k-\rho m^t},r}$
            \EndFor
        \EndFor
    \EndFor
    \State \textbf{Output} $w_{k,k}$ as~$y_k$
\end{algorithmic}
\end{algorithm}

At round~$t$, processor~$\rmP_k$ sends a message to processor~$\rmP_{k+\rho m^\tau}$, and receives a message from processor~$\rmP_{k-\rho m^\tau}$ through the~$\rho$-th port for~$\rho\in[p]$.
The sent message contains~$w_{k,r}$ for every~$r\in\cT_{k+\rho m^\tau}^{(t)}$.
The received message contains~$w_{{k-\rho m^\tau},r}$ for every~$r\in\cT_k^{(t)}$, and processor~$\rmP_k$ updates its memory, letting~$w_{k,r}\gets w_{k,r}+w_{{k-\rho m^\tau},r}$. 
The details are given in Alg.~\ref{alg:sp}, whose correctness is given as follows.
\begin{lemma}\label{lemma:shoot}
After the shoot phase, every processor~$\rmP_k$ has obtained~$y_k\triangleq\sum_{s\in\cS_k^-}w_{s,k}$, with
\begin{equation*}
    C_{1,\text{shoot}}=T_s,C_{2,\text{shoot}}= \frac{(p+1)^{T_s}-1}{p}.
\end{equation*}
\end{lemma}
\begin{proof}
We show that the variable~$w_{k,k}$ stores~$y_k$ for every~$k$ at the end of the algorithm.
The proof is based on an recursively defined tree~$\cT_k'$: in step~$\tau=0$ processor~$\rmP_k$ is added as the root. In every subsequent step~$\tau\ge 1$ every processor~$\rmP_r$ in the current tree is connected to processor~$r-\rho m^\tau$, for every~$\rho\in[p]$.
That is, during each step~$\tau$, we add processors that \textit{sent} messages to the existing processors at round~$t=T_s-\tau+1$.

Observe that the tree~$\cT_k'$ contains all processors in~$\cS_k^-$ after~$T_s$ steps of the recursion.
For each processor~$r\in\cT_k'$ added in step~$\tau$, processor~$\rmP_k$ is present in the tree~$\cT_r^{(t)}$.
Hence, at round~$t$, the packet~$w_{r,k}$ is transmitted from processor~$\rmP_r$ to its parent~$s$ in~$\cT_k'$, and summed with~$w_{s,k}$. 

Traversing from the processors added at step~$T_s$, the packets~$\{w_{r,k}\}_{r\in\cT_k'}$ are summed and transmitted to the root~$k$, and stored in the variable~$w_{k,k}$.
Since tree~$\cT_k'$ contains exactly the processors in~$\cS_k^-$, variable~$w_{k,k}$ stores~$y_k=\sum_{s\in\cS_k^-}w_{s,k}$.

Finally, observe that~$|\cT_k^{(t)}|=\frac{n}{(p+1)^t}$, which is the number of packets (elements in~$\bbF_q$) sent by processor~$\rmP_k$ at round~$t$ through each of its ports. Hence, summing over all rounds,
\small\begin{equation*}
     C_{2,\text{shoot}} =\sum_{t=1}^{T_s}(p+1)^{T_s-t}=\frac{(p+1)^{T_s}-1}{p}.\qedhere
\end{equation*}
\end{proof}\normalsize

Finally, in the most general case where~$K<mn$, some overlap of indices need to be resolved, as some computation results have been summed up twice.
In particular, observe that~$\cR_k^-\cap \cR_{k-(n-1)m}^-=[k-nm+1,k]$, which is an empty set only if~$nm=K$, as we assumed that~$(n-1)m<K\le nm$, and since indices are computed~$\bmod~K$.
Therefore,
\begin{equation}\label{eq:xkFromyk}
  \textstyle y_k=\coded{x}_k+\sum_{r\in[k-mn+1,k]}\bfC_{r,k}x_r,
\end{equation}
from which processor~$\rmP_k$ can individually compute~$\coded{x}_k$ with no communication, by computing the r.h.s sum and subtracting from~$y_k$.
This concludes the prepare-and-shoot algorithm, and provides the following by Lemma~\ref{lemma:prepare} and Lemma~\ref{lemma:shoot}.
\begin{theorem}\label{theorem:universal-c}
    The prepare-and-shoot algorithm has a communication cost of
    \small\begin{equation*}
    \begin{split}
        C&_\text{A2A,Univ}=\\
        &\begin{cases}
        \alpha \ceil{\log_{p+1}K} + \beta\ceil{\log_2q}\frac{(p+1)^{(L-1)/2}(p+2)-2}{p} & \mbox{if~$L$ is odd,}\\ 
        \alpha\ceil{\log_{p+1}K} + \beta\ceil{\log_2q}\frac{2(p+1)^{L/2}-2}{p} & \mbox{if~$L$ is even.}
        \end{cases}
    \end{split}
    \end{equation*}\normalsize
\end{theorem}

\begin{remark}\label{remark:approx-cost}
According to Lemma~\ref{lemma:universalC1bound}, the prepare-and-shoot algorithm is strictly optimal in terms of~$C_1$.
In addition, since~$(p+1)^{L}<K$, the algorithm has~$C_2\approx 2\sqrt{K}/p$, which is suboptimal in terms of~$C_2$ within a factor of~$\sqrt{2}$ (Lemma~\ref{lemma:universalC2bound}).
\end{remark}

\section{Algorithm for Vandermonde matrices}\label{section:vandermonde}
In this section, we shift our attention to specific all-to-all encode algorithms for computing Vandermonde matrices.
The procedure discussed in this section may be directly utilized in solving the  decentralized encoding problem for~\emph{non-systematic} Reed-Solomon codes (Appendix~\ref{appendix:non-sys}).
Otherwise, it may be utilized as a primitive in performing all-to-all encode operations for Cauchy-like matrices for their prevalent use in~\emph{systematic} Reed-Solomon codes (Section~\ref{section:reedsolo}).

We tailor both the scheduling and the coding scheme for square Vandermonde matrices~$\bfA\in\bbF_q^{K\times K}$ such that~$\bfA_{i,j}=\alpha_{j}^i$ for~$i,j\in[0,K-1]$, where~$\alpha_0,\ldots,\alpha_{K-1}$ are distinct elements of~$\bbF_q$ (hence~$K\le q$). 
That is, every processor~$\rmP_k$ requires~$\coded{x}_k=f(\alpha_k)$, where
\begin{equation}\label{eq:root-poly}
    f(z)=x_{0}\cdot z^0+x_{1}\cdot z^1+\cdots+x_{(K-1)}\cdot z^{K-1}.
\end{equation}

Inspired by the fast Fourier transform algorithm, we show a method that computes the~\emph{Discrete Fourier Transform} (DFT) matrix whose columns are permuted.
Later, this method serves as a primitive for the computation of general Vandermonde matrices, and brings a significant gain in~$C_2$ compared with the universal algorithm described earlier.

\subsection{Computing a (permuted) DFT Matrix }\label{section:DFT}
 
Assume that~~$K\vert q-1$ and $K=P^{H}$, and hence we may represent every~$k\in[0,K-1]$ as
\begin{equation}\label{eq:define-k}
    k=k_{H}P^{H-1}+k_{H-1}P^{H-2}+\cdots+k_{1}
\end{equation}
with~$H$ many integers, where $k_H,\ldots,k_1\in[0,P-1]$.
Meanwhile, we define the \textit{reverse} of~$k$ as
\begin{align}\label{eq:define-k-prime}
    k'&\triangleq k_1P^{H-1} + k_2P^{H-2} + \cdots+k_{H}\nonumber\\
    &= k'_{H}P^{H-1}+k'_{H-1}P^{H-2}+\cdots+k'_{1},
\end{align}
i.e., ${k'}$ is represented by the same set of integers as~$k$ but in the reversed order, and for~$h\in[1,H]$,~$k'_h=k_{H+1-h}$.

Let~$\beta=g^\frac{q-1}{K}$ be a primitive~$K$-th root of unity, where~$g$ is a generator of~$\bbF_q$.
A DFT matrix~$\bfD_K$ is a Vandermonde one with~$\alpha_k=\beta^{k}$, i.e.,
\small\begin{equation}\label{eq:DFTmatrix}
\bfD_K=\left[\begin{array}{ccccc}
    1      &1                &1              &\cdots & 1                 \\
    1      &\beta^{1}            &\beta^{2}        &\cdots & \beta^{(K-1)}       \\
    \vdots &\vdots           &\vdots         &\ddots & \vdots            \\
    1      &\beta^{{1}(K-1)}   &\beta^{2(K-1)} &\cdots & \beta^{(K-1)(K-1)}    \\
    \end{array}\right].
\end{equation}\normalsize

However, instead of computing~$\bfD_k$, our algorithm computes a slightly different matrix~$\bfD_k'\triangleq\bfD_k\bfP$, where~$\bfP$ is a permutation matrix such that~$\bfP_{k,k'}=1$ for every~$k\in[0,K-1]$, and~$0$ elsewhere.
This is because every processor~$\rmP_k$ will obtain an evaluation of a polynomial~$f(z)$~\eqref{eq:root-poly} at~$\beta^{k'}$ instead of~$\beta^k$ at the end of the algorithm.

To accomplish this task, we propose a method based on two trees; a tree of field elements and a tree of polynomials.
The former tree is defined recursively in a top–down approach.
Initially, the tree contains only the root vertex at level~$0$, which hosts~$\gamma=1$, i.e., the unity of the field~$\bbF_q$.
For~$h\in[H]$, level~$h$ of the tree is constructed by adding~$P$ child vertices to every existing vertex at level~$h-1$.
Specifically, consider some vertex that resides at level~$h-1$ and is indexed by
$
k_{h-1}k_{h-2}\cdots k_1.
$
Then, for every~$\rho\in[0,P-1]$, its~$\rho$-th child is indexed with~$\rho$ prepended to its own index, i.e.,
$
\rho k_{h-1}k_{h-2}\cdots k_1.
$

In other words, indices of sibling vertices differ only in the most significant (leftmost) digit, and the index of a parent vertex is given by omitting the leftmost digit of its child.
That is, 
\begin{equation*}
    {0k_{h-1}\cdots k_{0}},{1k_{h-1}\cdots k_{0}},\ldots,{(P-1)k_{h-1}\cdots k_{0}}
\end{equation*}
all share the same parent vertex indexed by~${k_{h-1}\cdots k_{0}}$.

Moreover, the vertex~$k_hk_{h-1}\cdots k_1$ hosts a field element
\begin{equation}\label{eq:define-field-element}
    \gamma_{k_{h}\cdots k_{1} } \triangleq (\beta^{k_{h}P^{h-1}+k_{h-1}P^{h-2}+\cdots+k_{1}})^{K/{P}^h}.\\
\end{equation}
Then, it is readily verified that every child element is a distinct~$P$-th root of the parent element, i.e., for~$\rho\in[0,P-1]$
\begin{align}\label{eq:treeVertexGamma}
    (\gamma_{\rho k_{h-1}\cdots k_{1}})^{P}
    &=(\beta^{\rho P^{h-1}+k_{h-1}P^{h-2}+\cdots+k_{1}})^{K/{P}^h\cdot{P}}\nonumber\\
    &=(\beta^{k_{h-1}P^{h-2}+\cdots+k_{1}}\cdot \beta^{\rho P^{h-1}})^{K/P^{h-1}}\nonumber\\
    &\overset{\eqref{eq:define-field-element}}{=}(\beta^{k_{h-1}P^{h-2}+\cdots+k_{1}})^{K/P^{h-1}}(\beta^{\rho})^{K}\nonumber\\
    &=\gamma_{k_{h-1}\cdots k_{1}}\cdot(g^\frac{q-1}{K})^{\rho K}=\gamma_{k_{h-1}\cdots k_{1}}.
\end{align}
In particular, each leaf (vertex at level~$H$) hosts
$$
\gamma_{k_{H}\cdots k_{1}}= \beta^{k_{H}P^{H-1}+k_{H-1}P^{H-2}+\cdots+k_{1}}=\beta^{k},
$$
which is the evaluation point of the processor~$\rmP_{k'}$, and the root (vertex at level~$0$) is represented by~$0$ digits and hosts~$\gamma=1$.

\begin{figure}
    \centering
    \includegraphics[width=0.48\textwidth]{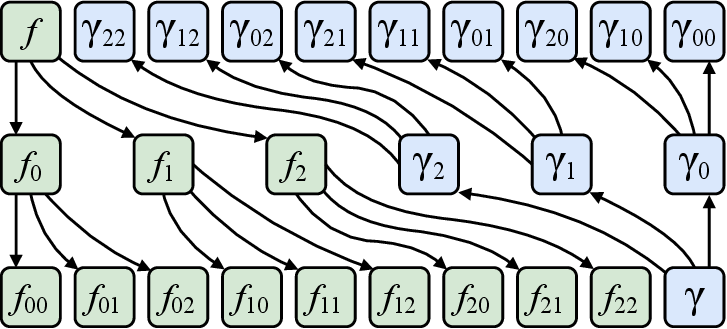}
    \caption{An illustration of the trees defined in \eqref{eq:treeVertexGamma} and~\eqref{eq:combinationOfChildren} for the special case~$K=9$ and~$p=2$.
    The polynomial tree (left) is rooted at $f(z)=x_0+x_1z+\cdots+x_{8}z^{8}$. 
    The polynomials in level~$1$ are 
    $f_0(z)=x_0+x_3z+x_{6}z^{3}$,
    $f_1(z)=x_1+x_4z+x_{7}z^{3}$, and~$f_2(z)=x_2+x_5z+x_{8}z^{3}$.
    The polynomials in level~$2$ are~$f_{00}=x_0,f_{01}=x_3,f_{02}=x_6,f_{10}=x_1,f_{11}=x_4,f_{12}=x_7,f_{20}=x_2,f_{21}=x_5,f_{22}=x_8$.
    The tree of field elements (right) is rooted at~$\gamma=1$ and every child of is a distinct cubic root of its parent.
    }
    \label{fig:tree}
\end{figure}

The tree of polynomials, on the contrary, is defined recursively in a bottom-up approach.
In this tree, every leaf vertex with index~$k_1k_{2}\cdots k_H$ hosts the constant polynomial
\begin{equation}\label{eq:leafVertex}
    f_{k_1k_2\cdots k_H}(z)=x_{k_{H}P^{H-1}+k_{H-1}P^{H-2}+\cdots+k_{1}}=x_{k}.
\end{equation}
Moreover, a parent vertex is indexed by omitting the leftmost digit of its children, and hosts a combination of their polynomials evaluated at~$z^{P}$, i.e.,
\begin{equation}\label{eq:combinationOfChildren}
   \textstyle f_{k_{1}\cdots k_{h-1}}(z)=\sum_{\rho\in[0,P-1]}z^{\rho} f_{k_{1}\cdots k_{h-1}\rho} (z^{P}).
\end{equation}
 
By~\eqref{eq:leafVertex} and~\eqref{eq:combinationOfChildren}, the polynomial at the root level is of the form~$\sum_{k\in[0,K-1]}x_kz^{b_k}$, and
\small\begin{align*}
        b_{k}&= k_1 + P(k_2 + P(k_3+\cdots +P (k_{H-1} +Pk_H)\cdots))\\
            &\overset{(a)}{=}k_1+k_2P\cdots+k_{H-1}P^{H-2}+k_{H}P^{H-1}\overset{\eqref{eq:define-k}}{=}k,
\end{align*}\normalsize
where~$(a)$ follows from Horner's rule.
As a result, the root polynomial tree is indeed~$f(z)$ as defined in~\eqref{eq:root-poly}.
See Figure~\ref{fig:tree} for examples of both trees.

The proposed algorithm is defined recursively using the above trees. 
For~$k,k'$ as defined in~\eqref{eq:define-k} and~\eqref{eq:define-k-prime}, respectively, let
$$
Q(k,h)\triangleq  f_{k_{H}\cdots k_{h}}(\gamma_{k_{h-1}\cdots k_{1}}),
$$
and hence every processor~$\rmP_k$ initially has
$$
x_k=  f_{k_{1}\cdots k_{H}} (\gamma) =f_{k'_{H}\cdots k'_{1}}(\gamma)= Q(k',H),
$$
which is an evaluation of~$f_{k_{1}\cdots k_{H}}(z)$ (a leaf in the polynomial tree) at~$\gamma=1$ (the root of the element tree).
It requires
$$
\coded{x}_k=f(\beta^{k'})=f(\gamma_{k'_{H}\cdots k'_{1}})=Q(k',H),
$$
which is an evaluation of~$f(z)$ (the root of polynomial tree) at~$\gamma_{k'_{H}\cdots k'_{1}}=\beta^{k'}$ (a leaf of the element tree).

Assume that every processor~$\rmP_{k}$ has learnt~$Q(k',h)$,
and we show how it can obtain~$Q(k',h+1)$ with one all-to-all encode operation.
Note that,
\small\begin{align}\label{eq:linearCombination}
   Q(k',h+1)=&f_{k'_{H}\cdots k'_{h+1}}(\gamma_{k'_{h}\cdots k'_{1}}) \nonumber\\
    \overset{\eqref{eq:combinationOfChildren}}{=}& \textstyle \sum_{\rho\in[0,P-1]}(\gamma_{k'_{h}\cdots k'_{1}})^{\rho} f_{k'_{H}\cdots k'_{h+1}\rho}((\gamma_{k'_{h}\cdots k'_{1}})^{P})\nonumber\\
    \overset{\eqref{eq:treeVertexGamma}}{=}&\textstyle \sum_{\rho\in[0,P-1]}(\gamma_{k'_{h}\cdots k'_{1}})^{\rho} f_{k'_{H}\cdots k'_{h+1}\rho}(\gamma_{k'_{h-1}\cdots k'_{1}})\nonumber\\
    =&\textstyle\sum_{\rho\in[0,P-1]}(\gamma_{k'_{h}\cdots k'_{1}})^{\rho} Q( \xi_{h,\rho}(k') ,h),
\end{align}\normalsize
where~$\xi_{h,\rho}(k)\triangleq~k+ (\rho-k_h)\cdot P^{h-1}$.
Therefore, the required~$Q(k',h+1)$ can be computed as a linear combination of
$$Q( \xi_{h,0}(k') ,h),\ldots,Q( \xi_{h,P-1}(k') ,h),$$
according to~\eqref{eq:linearCombination}.
For all processors~$\rmP_{\xi_{h,0}(k)},\ldots,\rmP_{\xi_{h,P-1}(k)}$, we write~\eqref{eq:linearCombination} in matrix form:
 \small\begin{equation*}
    \left[\begin{array}{c}
    Q( \xi_{h+1,0}(k'),h+1)\\
    \vdots\\
    Q( \xi_{h+1,P-1}(k'),h+1)
    \end{array}\right]^\intercal=
    \left[\begin{array}{c}
    Q( \xi_{h,0}(k'),h)\\
    \vdots\\
    Q( \xi_{h,P-1}(k'),h)
    \end{array}\right]^\intercal
    \cdot
    \bfA_k^{(h)},
\end{equation*} \normalsize 
where~$\bfA_k^{(h)}$ is a Vandermonde matrix defined by~$P$ field elements~$\gamma_{1k_{h-1}\cdots k_{0}},\ldots,\gamma_{Pk_{h-1}\cdots k_{0}}$ , i.e.,
\small\begin{equation}\label{eq:linearCombinationMatrix}
    \bfA_k^{(h)}=
    \left[\begin{array}{ccc}
    1 & \cdots & 1  \\
    \gamma_{1k'_{h-1}\cdots k'_{0}} & \cdots & \gamma_{Pk'_{h-1}\cdots k'_{1}}  \\
    \vdots                            & \ddots & \vdots\\
    (\gamma_{1k'_{h-1}\cdots k'_{0}})^{P-1} & \cdots & (\gamma_{Pk'_{h-1}\cdots k'_{1}})^{P-1} \\
    \end{array}\right].
\end{equation}\normalsize
Hence, processor~$\rmP_k$ can obtain~$Q(k',h+1)$ by using~$Q(k',h)$ as initial data, and performing an all-to-all encode operation on~$\bfA_k^{(h)}$ with processors that are in possession of
\begin{equation*}
    Q({\xi_{h,0},(k'),h}),\ldots,Q({\xi_{h,P-1},(k')},h),
\end{equation*}
i.e., the processors with the same index except for the~$h$-th digit.
This can be completed using the prepare-and-shoot algorithm introduced in Section~\ref{section:universal}.

Recall that every processor~$\rmP_k$ has~$Q(k',0)$ at the beginning of the algorithm, and processors in possession of~$Q(k',h)$ can obtain~$Q(k',h+1)$ by completing the above operation.
Therefore, by induction,~$\rmP_k$ obtains the coded packet~$\coded{x}_k=Q(k',H)$ after~$H$ rounds by repeating the above operation~$H$ times.
Hence, let~$C_\text{A2A,Univ}(P)$ be the communication cost in the prepare-and-shoot algorithm to compute any~$P\times P$ matrix stated in Theorem~\ref{theorem:universal-c}, and we have the following theorem.

\begin{theorem}\label{theorem:dft-cost}
    When~$K\mid q-1$ and~$K=P^H$ for some integers~$P,H$, the above algorithm for computing a (permuted) DFT matrix~$\bfD_k'=\bfD_k\bfP$ has
    \begin{equation*}
      C_\text{A2A,DFT}=H\cdot C_\text{A2A,Univ}(P).  
    \end{equation*}
\end{theorem}

\noindent As stated in Remark~\ref{remark:approx-cost}, the communication cost
\begin{equation*}
    C_\text{A2A,DFT}\approx H\cdot\left(\alpha\cdot\ceil{\log_{p+1} P }+\beta\ceil{\log_2q} \cdot \frac{2\sqrt{P}}{p}\right).
\end{equation*} 
Moreover, in the special cases that~$P=p+1$, every all-to-all encode operation may be completed with one~$1$-to-all broadcast.
Therefore, the communication cost is strictly optimal (see Remark~\ref{remark:bounds4SpecificMatrices}), as stated in the following corollary.
\begin{corollary}\label{theorem:optimalityOfDrawAndLoose}
    When~$K\mid q-1$ and~$K=(p+1)^H$, the above algorithm for computing a (permuted) DFT matrix has
    \begin{equation*}
    \textstyle C_\text{A2A,DFT}= H\cdot(\alpha+\beta\ceil{\log_2q} ).
    \end{equation*}
\end{corollary}

Next, we emphasize the~\emph{invertibility} of the presented algorithm in the following lemma; this will be useful in the computation of Cauchy-like matrices.
\begin{lemma}\label{lemma:invertibilityOfLoose}
    The above algorithm can be used to compute the inverse of a (permuted) DFT matrix, with the same~$C_1$ and~$C_2$.
\end{lemma}
\begin{proof}
Since the matrix~$\bfA_k^{(t)}$ defined in~\eqref{eq:linearCombinationMatrix} is an invertible Vandermonde matrix, it follows that each step of the induction is invertible.
That is, the processor~$\rmP_k$ in possession of~$Q(k',h+1)$ can obtain~$Q(k',h)$ with one all-to-all encode operation on the inverse of~$\bfA_k^{(h)}$.
Therefore, a processor~$\rmP_k$ in possession of~$Q(k',H)$ can obtain~$Q(k',0)$ in~$H$ rounds, with the same~$C_1$ and~$C_2$ as shown in Theorem~\ref{theorem:optimalityOfDrawAndLoose}.
\end{proof}

\subsection{Generalization to Vandermonde matrices}

The aforementioned algorithm computes a unique Vandermonde matrix with improved communication cost, but requires that~$K=P^H$ for some integers~$P,H$ and that~$K\vert q-1$. 
In cases where~$K\nmid q-1$ and~$K\le q-1$,
we use the above DFT algorithm as a primitive for improved~$C_2$ with respect to prepare-and-shoot in the computation of multiple other Vandermonde matrices, described next.

For some integer~$P$, let~$H$ be the maximum integer such that~$Z=P^H$ divides both~$K$ and~$q-1$ (i.e.,~$P^H\vert\gcd(K,q-1)$).
Let~$K=M\cdot Z$, and we organize processors in a~$M\times Z$ grid in which processor~$\rmP_{i,j}=i\cdot Z+j$ resides in row~$i$ and column~$j$, for~$i\in[0,M-1]$ and~$j\in[0,Z-1]$.

Let~$\varphi$ be a injective map from~$[0,M-1]$ to~$[0,(q-1)/Z-1]$.
Note that such a mapping exists since~$q-1\ge ZM=K$. 
We define the evaluation point of processor~$\rmP_{i,j}$ as
\begin{equation}\label{eq:omega-ij}
    \omega_{i,j}\triangleq \alpha_i\cdot\beta_{j'},~\mbox{where}~\alpha_i=g^{\varphi(i)}~\mbox{and}~\beta_{j'}=  g^{{j'}\cdot\frac{q-1}{Z}}.
\end{equation}
The injectivity of~$\varphi$ assures that no two evaluation points are identical, and hence no two processors require the same coded packet\footnote{Since there exists~$\binom{(q-1)/Z}{M}$ possible choices for~$\varphi$, it follows that the proposed algorithm computes this many different Vandermonde matrices up to permutation of columns.}.
Note that~$j'\in[0,Z-1]$ is the reverse of~$j$, similarly defined as~$k'$ in~\eqref{eq:define-k-prime}.

Recall that the coded data~$\coded{x}_{i,j}$ required by processor~$\rmP_{i,j}$ is an evaluation of~$f(z)$~\eqref{eq:root-poly} at~$\omega_{i,j}$.
We have
\small
\begin{align}\label{eq:extract-fl-1}
   \coded{x}_{i,j}=f(\omega_{i,j})&=\textstyle  \sum_{k=0}^{K-1}x_k(\omega_{i,j})^k\overset{\eqref{eq:omega-ij}}{=}\sum_{k=0}^{K-1}x_k\alpha_i^k\beta_{j'}^k\nonumber\\
    &\textstyle =\sum_{\ell=0}^{Z-1} \beta_{j'}^\ell \left(\sum_{k=\ell\bmod Z}x_k\alpha_{i}^k\beta_{j'}^{k-\ell}\right). 
\end{align}
\normalsize
Since~$g^{q-1}=1$, it follows that~$\beta_{j'}^{k-\ell}=g^{j'\cdot\frac{q-1}{Z}(k-\ell)}=1$ whenever~$k=\ell\bmod Z$, and hence
\small\begin{align}
    \coded{x}_{i,j}&\overset{\eqref{eq:extract-fl-1}}{=}\sum_{\ell=0}^{Z-1}\beta_{j'}^\ell\sum_{k=\ell\bmod Z}x_k\alpha_{i}^k=\sum_{\ell=0}^{Z-1}\beta_{j'}^\ell f_\ell(\alpha_i),~\mbox{where}\label{eq:extract-fl-2}\\
    f_\ell(z)&= x_{\ell}z^{\ell}+x_{\ell+Z}z^{\ell+Z}+\cdots+x_{\ell+Z\cdot (M-1)}z^{\ell+Z\cdot (M-1)}\nonumber\\
    &=x_{0,\ell}z^{\ell}+x_{1,\ell}z^{\ell+Z}+\cdots+x_{M-1,\ell}z^{\ell+Z\cdot (M-1)}.\label{eq:inside-of-fl}
\end{align}\normalsize

In matrix form, for every~$i\in[0,M-1]$,
\small\begin{equation}\label{eq:loosePhase}
    \begin{bmatrix}
     \coded{x}_{i,0}\\
        \vdots\\
        \coded{x}_{i,Z-1}
    \end{bmatrix}^\intercal
    \overset{\eqref{eq:extract-fl-2}}{=}
    \begin{bmatrix}
        f_0(\alpha_i)\\
        \vdots\\
        f_{Z-1}(\alpha_i)
    \end{bmatrix}^\intercal
    \cdot
    \left[ {\begin{array}{ccc}
    \beta_{0}^0    &  \cdots & \beta_{Z-1}^0    \\
    \vdots       &  \ddots & \vdots           \\
    \beta_{0}^{Z-1}&  \cdots & \beta_{Z-1}^{Z-1}\\
    \end{array} } \right]\cdot \bfP,
\end{equation}\normalsize
in which~$\bfP$ is a permutation matrix such that~$\bfP_{j,j'}=1$ for every~$j\in[0,Z-1]$, and~$0$ elsewhere.
Moreover, define
\small\begin{equation}\label{eq:define-vm}
\bfV_M\triangleq
\left[ {\begin{array}{ccc}
    \alpha_0^{Z\cdot0}  & \cdots & \alpha_{M-1}^{Z\cdot0}\\
    \vdots              & \ddots & \vdots \\
    \alpha_0^{Z(M-1)}   & \cdots & \alpha_{M-1}^{Z(M-1)} 
    \end{array} } \right],
\end{equation}\normalsize
and hence for every~$j\in[0,Z-1]$, we have
\small\begin{equation}\label{eq:drawPhase}
    \left[ {\begin{array}{c}
        f_j(\alpha_0)\\
        \vdots\\
        f_j(\alpha_{M-1})
    \end{array} } \right]^\intercal
    \overset{\eqref{eq:inside-of-fl}}{=}
    \begin{bmatrix}
        x_{0,j}\\
        \vdots\\
       x_{M-1,j}
    \end{bmatrix}^\intercal
    \cdot \bfV_M\cdot\diag
    \left[ {\begin{array}{c}
    \alpha_0^j \\
    \vdots     \\
     \alpha_{M-1}^j
    \end{array} } \right]^\intercal.
\end{equation}\normalsize

The above observations allow us to decompose the procedure into the~\emph{draw} and the~\emph{loose} phase, detailed as follows.

\textbf{Draw Phase:}
As shown in~\eqref{eq:drawPhase},~$(f_j(\alpha_0),\ldots,f_j(\alpha_{M-1}))$ is given by multiplying~$(x_{0,j},\ldots,x_{M-1,j})$ with Vandermonde matrix~$\bfV_M$ defined by~$\{\alpha_0^Z,\ldots,\alpha_{M-1}^Z\}$, followed by matrix~$\diag(\alpha_0^j,\cdots,\alpha_{M-1}^j)$.
Hence, processors in column~$j$ may compute matrix~$\bfV_m$ using the universal prepare-and-shoot (Section~\ref{section:universal}), and once completed, every processor~$\rmP_{i,j}$~\emph{locally} multiplies the resulting packet with~$\alpha_i^j$ and obtains~$f_j(\alpha_i)$

\textbf{Loose Phase:}
As shown in~\eqref{eq:loosePhase},~$\coded{x}_{i,0},\ldots,\coded{x}_{i,Z-1}$ are linear combinations of polynomial evaluations~$f_0(a_i),\ldots,f_{Z-1}(a_i)$, which are defined by the columns of the (permuted) DFT matrix~$\bfD_Z\bfP$.
Therefore, this phase involves~$M$ row-wise all-to-all encode operations.
For~$i\in[0,M-1]$, processors in row~$i$ compute~$\bfD_Z\bfP$ using the specialized algorithm a for the permuted DFT matrix (Section~\ref{section:DFT}).
After~$H$ rounds, every processor~$\rmP_{i,j}$ obtains the coded packet~$\coded{x}_{i,j}$.

Observe that the draw phase takes $C_1=\ceil{\log_{p+1}M}$ rounds and~$C_2=C_\text{A2A,Univ}(M)$ communication, and the loose phase takes~$C_1=H$ rounds and~$C_2=H$ communication. 
Therefore, we have the following.

\begin{theorem}\label{theorem:Vandermonde}
The presented draw-and-loose algorithm can compute~$\binom{(q-1)/Z}{M}$ different Vandermonde matrices (up to permutation of columns) with
\begin{equation*}
    C_\text{A2A,Vand}=(\alpha+\beta \ceil{\log_2q})\cdot H +  C_\text{A2A,Univ}(M).
\end{equation*}
In particular, if~$M\le P^2$ then~$C_\text{A2A,Univ}(M)=2(\alpha+\beta )$, i.e.,
\begin{equation*}
    C_\text{A2A,Vand}=(\alpha+\beta \ceil{\log_2q})\cdot \ceil{\log_{p+1} Z}.
\end{equation*}
\end{theorem}
\begin{remark}
    Note that draw-and-loose can compute any Vandermonde matrix. 
    Yet, significant gains w.r.t prepare-and-shoot will be possible in cases where~$H$ is large.
    In particular, in cases where~$H=0$ the draw-and-loose algorithm does not provide gains over the universal prepare-and-shoot.
\end{remark}
Similar to Lemma~\ref{lemma:invertibilityOfLoose}, the above algorithm is invertible.
\begin{lemma}\label{lemma:invertibilityOfDrawLoose}
    Draw-and-loose can compute the inverse of a Vandermonde matrix, with the same communication cost.
\end{lemma}
\begin{proof}
The invertibility of the loose phase is given in Lemma~\ref{lemma:invertibilityOfLoose}.
In the draw phase, since the Vandermonde matrix~$\bfV_M$ defined in~\eqref{eq:define-vm} is invertible, this step can be inverted by computing the inverse of~$\bfV_M$ using prepare-and-shoot. 
Together, the inverse of a Vandermonde matrix can be computed by first inverting the loose phase, and computing the inverse of~$\bfV_M$ with prepare-and-shoot.
\end{proof}

\section{Decentralized Encoding for systematic Reed-Solomon codes}\label{section:reedsolo}

In this section, we show that the decentralized encoding of a systematic Reed-Solomon code can be completed by performing two consecutive draw-and-looses.
This fact allows further optimization of the decentralized encoding process.

As discussed in~\cite[Ch.~5.1.2]{roth2006coding}, an~$[N=K+R,K]$ generalized Reed Solomon (GRS) code has a generator matrix
\small\begin{equation}\label{eq:gen-matrix-grs}
\bfG_\text{GRS}=
  \left[\begin{array}{c|c} \bfV_\alpha&\bfV_\beta\end{array}\right]
\cdot
\diag(u_0,\ldots,u_{K-1},v_{0},\ldots,v_{R-1}),
\end{equation}\normalsize
\small
\begin{equation*}
\mbox{with}~
    \bfV_\alpha=
    \begin{bmatrix}
        1 & \ldots & 1 \\
        \alpha_0 & \ldots & \alpha_{K-1}\\
        \vdots & \ddots & \vdots \\
        \alpha_0^{K-1} & \ldots & \alpha_{K-1}^{K-1}\\
    \end{bmatrix},
    \bfV_\beta=\begin{bmatrix}
        1 & \ldots & 1 \\
        \beta_0 & \ldots & \beta_{R-1}\\
        \vdots & \ddots & \vdots \\
        \beta_0^{K-1} & \ldots & \beta_{R-1}^{K-1}\\
    \end{bmatrix}
\end{equation*}\normalsize
being two Vandermonde matrices defined by distinct elements~$\alpha_0,\ldots,\alpha_{K-1},\beta_{0},\ldots,\beta_{R-1}\in\bbF_q$, and~$u_0,\ldots,u_{K-1},$
$v_0,\ldots,v_{R-1}\in\bbF_q$ being non-zero (but not necessarily distinct) elements.

In coded distributed systems, the systematic property is usually desired as it allows the~$K$ sources to store the raw data.
The GRS code in~\eqref{eq:gen-matrix-grs} has a systematic form given by
\small\begin{equation}\label{eq:gen-matrix-systematic}
    \bfG_\text{SGRS}=\left[\begin{array}{c|c} \bfI&\bfA\end{array}\right],~\mbox{where}~\bfA=(\bfV_\alpha\cdot \bfP)^{-1}\cdot \bfV_\beta\cdot \bfQ,
\end{equation}\normalsize
with~$\bfP\triangleq\diag(u_0,\ldots,u_{K-1})~\mbox{and}~\bfQ\triangleq\diag(v_{0},\ldots,v_{R-1})$.
Note that by~\cite[Theorem~1]{roth1985generator},
the matrix~$\bfA$ in~\eqref{eq:gen-matrix-systematic} satisfies
\begin{equation}\label{eq:define-cauchy-like}
    \bfA_{k,r}= c_kd_r/(\beta_r-\alpha_k),~\mbox{where}
\end{equation}
\small\begin{equation*}
    c_k =  \frac{u_k^{-1}}{\prod_{t\in[0,K-1],t\neq k} (\alpha_k-\alpha_{t}) }~\mbox{and}~d_r =v_{r}{\prod_{k\in[0,K-1]}(\beta_r-\alpha_{k})}.
\end{equation*}\normalsize
Note that the submatrices~$\bfA_0,\ldots,\bfA_{M-1}$ as defined in~\eqref{eq:submatrix-of-A} or~\eqref{eq:submatrix-of-A-horizontal} are all of this form, which are referred~\emph{Cauchy-like}.

In the case of~$K\geq R$, we provide a decomposition of every submatrix~$\bfA_m$ in the following theorem, which enables to use draw-and-loose (Section~\ref{section:vandermonde}).

\begin{theorem}\label{theorem:cauchy}
    Let~$\bfA$ be the non-systematic submatrix of~$\bfG_\text{SGRS}$ as in~\eqref{eq:gen-matrix-systematic}, and let~$\bfA_m$ be the~$m$-th submatrix of~$\bfA$ as defined in~\eqref{eq:submatrix-of-A}.
    Let~$\bfV_{\alpha,m}\in\bbF_q^{R\times R}$ be a Vandermonde matrix defined by the vector~$(\alpha_{i})_{i\in\cS_m}$ with~$\cS_m\triangleq[mR,(m+1)R-1]$.
    Then,
    \begin{equation}\label{eq:v-identity-for-c}
        \bfA_m= (\bfV_{\alpha,m}\cdot \Phi_m )^{-1}\bfV_\beta\cdot\Psi,
    \end{equation}
    where~$\Phi_m=\diag(\varphi_{m,0},\ldots,\varphi_{m,R-1})$, with
    \small\begin{equation}\label{eq:define-varphi}
        \textstyle\varphi_{m,s}=u_{mR+s}\prod_{\substack{j\in[0,K-1],j\not\in\cS_m}}\left(\alpha_{mR+s}-\alpha_j\right),s\in[0,R-1],
    \end{equation}\normalsize
    and~$\Psi_m=\diag(\psi_0,\ldots,\psi_{R-1})$, with
    \small\begin{equation}\label{eq:define-psi}
    \textstyle\psi_r= v_r\prod_{\substack{j\in[0,K-1],j\not\in\cS_m}}(\beta_r-\alpha_j),r\in[0,R-1].
    \end{equation}\normalsize
\end{theorem}
\begin{proof}
Since~$\bfV_{\alpha,m}$ is invertible, it follows that for~$s,r\in[0,R-1]$, the element~$(\bfV_{\alpha,m}^{-1})_{s,r}=\ell_{s,r}$ is a coefficient of the Lagrange basis polynomial
\begin{align}\label{eq:lag-basis-1}
    \ell_s(\xi)&=\ell_{s,0}+\ell_{s,1}\xi+\cdots+\ell_{s,R-1}\xi^{R-1}\nonumber\\
             &\textstyle=\prod_{r\in[0,R-1],r\neq s}\frac{\xi-\alpha_{mR+r}}{\alpha_{mR+s}-\alpha_{mR+r}}.
\end{align}
Let~$i=mR+s\in\cS_m$ for some~$s\in[0,R-1]$ and~$m\in[0,M-1]$.
Similarly,~$L_{i,s}=(\bfV_\alpha^{-1})_{i,s}$ is a coefficient of the Lagrange basis polynomial
\small\begin{align}\label{eq:lag-basis-2}
    &L_{i}(\xi)=L_{i,0}+\cdots+L_{i,K-1}\xi^{K-1}=\prod_{\substack{j\in[0,K-1],j\neq i}}\frac{\xi-\alpha_j}{\alpha_i-\alpha_j}\nonumber\\
    &=\prod_{\substack{j\in\cS_m,j\neq i}}\frac{\xi-\alpha_j}{\alpha_i-\alpha_j}\prod_{\substack{j\in[0,K-1],j\not\in\cS_m}}\frac{\xi-\alpha_j}{\alpha_i-\alpha_j}\nonumber\\
    &=\prod_{\substack{r\in[0,R-1],r\neq s}}\frac{\xi-\alpha_{mR+r}}{\alpha_{mR+s}-\alpha_{mR+r}}\prod_{\substack{j\in[0,K-1],j\not\in\cS_m}}\frac{\xi-\alpha_j}{\alpha_i-\alpha_j}\nonumber\\
    &\overset{\eqref{eq:lag-basis-1}}{=}\ell_s(\xi) \prod_{{j\in[0,K-1],j\not\in\cS_m}}\frac{\xi-\alpha_j}{\alpha_i-\alpha_j} \nonumber\\
    &\textstyle=\left(\prod_{\substack{j\in[0,K-1]\\j\not\in\cS_m}}\frac{1}{\alpha_i-\alpha_j}\right)\ell_s(\xi)\left(\prod_{\substack{j\in[0,K-1]\\j\not\in\cS_m}}(\xi-\alpha_j)\right).
\end{align}\normalsize
 Clearly,~$[\bfV_{\alpha,m}^{-1}\bfV_\beta]_{s,r}=\ell_s(\beta_r)$ and~$(\bfV_{\alpha}^{-1}\bfV_\beta)_{i,r}=L_{i}(\beta_r)$.
Therefore, with~$\varphi_{m,s}$ and~$\psi_r$ are as defined in~\eqref{eq:define-varphi} and~\eqref{eq:define-psi}, respectively, it follows that
\begin{equation*}
    \begin{split}
        (\bfA_m)_{s,r}&\overset{\eqref{eq:gen-matrix-systematic}}{=} u_{mR+s}^{-1}\cdot (\bfV_{\alpha}^{-1}\bfV_\beta)_{mR+s,r}\cdot v_r\\
        &= u_{mR+s}^{-1}\cdot L_{mR+s}(\beta_r)\cdot v_r\overset{\eqref{eq:lag-basis-2}}{=}\varphi_{m,s}^{-1}\cdot\ell_s(\beta_r)\cdot\psi_r,
    \end{split}
\end{equation*}
and as a result,~$\bfA_m= (\bfV_{\alpha,m}\cdot \Phi_m )^{-1}\bfV_\beta\cdot\Psi$.
\end{proof}

Theorem~\ref{theorem:cauchy} allows computing~$\bfA_m$ by performing two consecutive all-to-all encode operations, described as follows.
For~$m\in[0,M-1]$ and~$s\in[0,R-1]$, every source processor~$\rmS_{s,m}$ first scales its initial data by~$\varphi_{m,s}^{-1}$, and obtains a new data packet~$\varphi_{m,s}^{-1}x_{mR+s}$.
Then, with the new data packets, source processors in~$\cS_m$ perform all-to-all encode to compute the~\emph{inverse} of~$\bfV_{\alpha,m}$; this is made possible by the invertibility of draw-and-loose 
(Lemma~\ref{lemma:invertibilityOfDrawLoose}).

Next, they perform all-to-all encode on~$\bfV_{\beta}$ using draw-and-loose, and the result is scaled by~$\psi_r$ for the required coded packet.
This yields the following theorem.
\begin{theorem}\label{theorem:cauchy-c}
    The above algorithm computes matrix~$\bfA_m$ with
    \small\begin{equation*}
        C=\alpha\cdot 2\ceil{\log_{p+1}R}+\beta\ceil{\log_2q} \cdot \left( C_2(\bfV_{\alpha,m})+C_2(\bfV_\beta)\right).
    \end{equation*}\normalsize
\end{theorem} 

In the case of~$K<R$, we provide a decomposition of~$\bfA_m$ that is similar to Theorem~\ref{theorem:cauchy}, and hence the proof is omitted.
\begin{theorem}\label{theorem:cauchy-1}
    Let~$\bfA$ be the non-systematic submatrix of~$\bfG_\text{SGRS}$ in~\eqref{eq:gen-matrix-systematic}, and let~$\bfA_m$ be the~$m$-th submatrix of~$\bfA$ as defined in~\eqref{eq:submatrix-of-A-horizontal}.
    Let~$\bfV_{\beta,m}\in\bbF_q^{K\times K}$ be a Vandermonde matrix defined by the vector~$(\beta_{i})_{i\in\cT_m}$, and define~$\bfQ_m\triangleq\diag((v_r)r\in\cT_m)$, with~$\cT_m\triangleq[mK,(m+1)K-1]$.
    Then,
    \begin{equation}\label{eq:v-identity-for-c-1}
        \bfA_m= (\bfP\cdot \bfV_{\alpha})^{-1}\bfV_{\beta,m}\cdot \bfQ_m.
    \end{equation}
\end{theorem}
Using a similar strategy, we have the following.
\begin{theorem}\label{theorem:cauchy-c-1}
    The above algorithm computes matrix~$\bfA_m$ with
    \small\begin{equation*}
        C=\alpha\cdot 2\ceil{\log_{p+1}K}+\beta\ceil{\log_2q} \cdot (C_2(\bfV_{\alpha})+C_2(\bfV_{\beta,m})).
    \end{equation*}\normalsize
\end{theorem}
 
Overall, with specially chosen~$\bfV_\alpha$ and~$\bfV_\beta$ that define the generator matrix of~$\bfG_\textbf{SGRS}$ in~\eqref{eq:gen-matrix-systematic}, every~$\bfA_m$ can be computed by two consecutive draw-and-loose operations.
Note that this method doubles the required number of rounds, and hence is suitable for systems with relatively small start-up time~$\alpha$.
Together with the general framework introduced in Section~\ref{section:framework}, and assume that performing all-to-all encode operation for the Cauchy-like matrix takes~$C_\text{CL}$ as given in Theorem~\ref{theorem:cauchy-c} or Theorem~\ref{theorem:cauchy-c-1}, the decentralized encoding for such a systematic Reed-Solomon code requires
\begin{equation*}
    C= C_\text{BR}(\ceil{\textstyle\frac{K}{R}})+C_\text{CL}.    
\end{equation*}

\begin{remark}\label{remark:lagrange}
    Lagrange matrices were recently popularized for their use in Lagrange coded computing (LCC)~\cite{yu2019lagrange}, which concerns the problem of evaluating a polynomial~$f$ on a dataset~$\{x_1,\ldots,x_{K-1}\}$ with~$N>K$ worker nodes.
    In a nutshell, the input data~$x_0,\ldots,x_{K-1}$ is used to interpolate a polynomial~$g$ at~$\alpha_0,\ldots,\alpha_{K-1}\in\bbF_q$, and hence~$x_k=g(\alpha_k)$.
    The coded data are evaluations of~$g$ at~$\beta_0,\ldots,\beta_{N-1}\in\bbF_q$, i.e.,~$\coded{x}_n=g(\beta_n)$.
    The worker nodes compute~$f(\coded{x}_0),\ldots,f(\coded{x}_{N-1})$, from which the actual results~$f(x_0),\ldots,f(x_{K-1})$ can be decoded.

    Hence, the coded data satisfies
    \begin{equation*}
        (\coded{x}_0,\ldots,\coded{x}_{N-1})=(x_0,\ldots,x_{K-1})\cdot\bfV_\alpha^{-1}\cdot\bfV_\beta,
    \end{equation*}
    where~$\bfV_\alpha$ and~$\bfV_\beta$ are Vandermonde matrices specified by the~$\alpha_i$'s and~$\beta_i$'s, respectively. 
    Define Lagrange matrix~$\bfL_{\alpha,\beta}=\bfV_\alpha^{-1}\cdot\bfV_\beta$; it is Cauchy-like as introduced in~\eqref{eq:define-cauchy-like} with~$u_k=v_r=1$ for every~$k\in[0,K-1]$ and~$r\in[0,R-1]$. 

    If~$\alpha_k=\beta_k$ for~$k\in[0,K-1]$, then~$\bfL_{\alpha,\beta}$ is a systematic one.
    Otherwise, it is non-systematic.
    Either way, the all-to-all encode algorithm for Cauchy-like matrices presented in this section can be employed in decentralized encoding collectively with the frameworks introduced in Section~\ref{section:framework} or Appendix~\ref{appendix:non-sys}.
\end{remark}

\section{Discussion and Future Work}\label{section:discussion}
In this paper, we studied the decentralized encoding problem in linear network model.
We provided communication-efficient solution to this problem based on a newly defined collective communication operation called all-to-all encode.
We provide a universal solution which applies to any systematic linear code, and for the widely used Reed-Solomon code, we provide a specific method with a significant gain in communication costs over the universal one. 
While we assumed that the communication links between any two processors are noiseless, our methods can be easily integrated into noisy environments.
In such environments, the processors apply some error-correcting code over their sent packets prior to sending them, and the received packets undergo the respective decoding process.
Future work includes extending our results to Reed-Muller codes, to storage-specific codes such as regenerating or locally recoverable codes, and applying computation and storage constraints on the processors.



\printbibliography

\begin{appendices}

\section{Implementation of broadcast/reduce}\label{app:broadcast}
The broadcast and reduce operations have been extensively investigated in the literature.
Based on structures ranging from simple binomial tree and binary tree to more complicated ones such as circulant graph~\cite{traff2008optimal}, hypercube~\cite{ho1991optimal}, Fibonacci Trees~\cite{bruck1992multiple}, or a combination of two trees~\cite{sanders2009two}, numerous algorithms have been proposed to optimize communication cost with respect to different port number ($1$- verses multi-), start-up overhead~$\alpha$, per-unit cost~$\beta$, and data size~$W$.

When~$W$ is small, the folklore $(p+1)$-nomial tree broadcast may be used, which incurs a communication cost of
\begin{equation*}
    C_\text{BR}(N,W)=(\alpha +\beta \ceil{\log_2q} W)\cdot \ceil{\log_{p+1}N}.
\end{equation*}
This algorithm is used in the implementation of the universal all-to-all encode operation in Section~\ref{section:universal}.
However, it is not suitable for broadcasting big data (i.e. vector with large~$W$).
In a nutshell, processors that have not yet obtained the data have to wait for delivery from the processors that have received; this process may be tedious since the data are passed as a whole, and during which the bandwidth is not fully utilized.

To this end, pipelined algorithms are introduced in which processors may segment the data into smaller packets and send each independently~\cite[Chapter 13.1.2]{sanders2019sequential}.
With some tree topology defined in which the broadcaster serves as the root, a pipelined algorithm is developed in the following way:
Every vertex forwards every received packet to its children, and the algorithm concludes after every leaf processor has received the last packet.
The performance of this design depends on the underlying tree structure; for a thorough introduction to this line of algorithms, see~\cite[Chapter 13.1.2]{sanders2019sequential}.

To the best of our knowledge, the optimal broadcast algorithm in the one-port model is given in~\cite{traff2008optimal}, which incurs a communication cost of
\begin{equation*}
\textstyle    C_\text{BR}(N,W)=\left(\sqrt{(\ceil{\log_2 N}-1)\alpha}+\sqrt{\beta\ceil{\log_2q} W}\right)^2. 
\end{equation*}
Meanwhile, the optimal multi-port broadcast algorithm is provided in~\cite{bar1999broadcasting}.
By segmenting the data vector into~$m$ packets, the communication cost is
\begin{equation*}
    C_\text{BR}(N,W)=\left(\ceil{\textstyle\frac{m}{p}}+ \ceil{\log_{p+1}N}  \right)\cdot\left(\alpha + \beta\cdot\ceil{\textstyle\frac{W\ceil{\log_2q}}{m}}\right).
\end{equation*}

For a complete examination of the performance of broadcast and reduce operations, see~\cite{pjevsivac2007performance}.

\section{Extension to non-systematic codes}\label{appendix:non-sys}
Although systematic codes are generally preferred in real-world scenarios for their simplicity, certain applications require the code to be non-systematic.
For example, LCC~\cite{yu2019lagrange} employs a non-systematic Lagrange matrix in order to prevent the worker node from learning the initial data.
As such, we extend the proposed decentralized encoding framework to non-systematic codes, in which the~$N=K+R$ coded data vectors
\begin{equation*}
    \left(\coded{\bfx}^\intercal_0 ,\ldots,\coded{\bfx}^\intercal_{N-1}\right) = \left(\bfx^\intercal_0,\ldots,\bfx^\intercal_{K-1}\right)\cdot \bfG
\end{equation*}
are computed from the initial data vectors using a non-systematic generator matrix~$\bfG\in\bbF_q^{K\times N}$.

\subsection{The case of \texorpdfstring{$K>R$}{K>R}}\label{section:app-case-1}
Given a generator matrix~$\bfG\in\bbF_q^{K\times N}$, we create~$\bfG'$ by placing an arbitrary matrix~$\bfB\in\bbF_q^{R\times K}$ at its bottom, i.e.,
\begin{equation*}
    \bfG'\triangleq\left[\begin{array}{ c  }
        \bfG  \\
    \hline
    \bfB
  \end{array}\right]\in\bbF_q^{N\times N}.
\end{equation*}

Next, we enable the~$R$ sink processors to hold zero data packets~$\bfx_K,\ldots,\bfx_{N-1}$.
Therefore, the coded packets are
\begin{equation*}
    \left(\coded{\bfx}_0^\intercal,\ldots,\coded{\bfx}_{N-1}^\intercal\right)=\left(\bfx_0^\intercal,\ldots,\bfx_{K-1}^\intercal,\mathbf{0}^\intercal,\ldots,\mathbf{0}^\intercal\right)\cdot \bfG'
\end{equation*}
and hence the decentralized encoding can be completed by one all-to-all encode operation for computing~$\bfG'$ among all~$N$ processors.
The communication cost is then
$
    C=C_\text{A2A}(\bfG').
$
Note that the choice of~$\bfB$ does not affect the decentralized encoding results since the sink processors hold~$\mathbf{0}$'s.

\begin{figure}
    \centering
    \scalebox{0.62}{\input{figures/example-framework-3}}
    \caption{An example of completing decentralized encoding for non-systematic matrix~$\bfG$ when~$p=1$,~$K=4$ and~$R=27$.
        Sinks~$\rmT_0,\ldots,\rmT_{23}$ are placed in a~$4\times 6$ grid, while sinks $\rmT_{24},\rmT_{25}$ and~$\rmT_{26}$ are stacked at the bottom of the first, the second, and the third column, holding empty packets.
        In phase one, sources broadcast their initial data to others in the same row.
        In phase two, processors in each column~$m$ perform all-to-all encode to compute~$\bfG'_m$.
    }
    \label{fig:framework-3}
\end{figure}

\subsection{The case of \texorpdfstring{$K\leq R$}{K>R}}

Let~$M$ be the least integer such that~$M\cdot K>R$, and let~$L=R\mod K$.
We treat the matrix~$\bfG\in\bbF_q^{K\times N}$ as a concatenation of square matrices~$\bfG_0,\ldots,\bfG_{M-1}\in\bbF_q^{K\times K}$ followed by a (not necessarily square, and potentially empty) matrix~$\bfG_{M}\in\bbF_q^{K\times L}$, i.e.,
$
    \bfG=[        \bfG_0 \cdots  \bfG_{M-1}~\bfG_{M}]. 
$

For convenience of indexing, we organize the~$R$ sinks in a~$K\times M$ grid.
For~$r=k+(m-1)\cdot K$, processor~$\rmT_r$, now renamed as $\rmT_{k,m}$, is placed in row~$k\in[0,K-1]$ and column~$m\in[M]$.
For ease of demonstration, the source processors are prepended as column~$0$ of the processor grid.
Next, we evenly distribute processors in the incomplete column~$M$ to other columns.
An illustrative example can be found in Figure~\ref{fig:framework-3}.

The proposed framework involves two phases described as follows.
In phase one, we initiate~$K$ instances of one-to-all broadcast in parallel, which allows each source to share its initial data with other processors
in the same row, i.e., for every~$k\in[0,K-1]$,~$\bfx_k$ is disseminated to processor~$\rmT_{k,m}$ for~$m\in[M-1]$.

In phase two, for~$m\in[M-1]$, processors in column~$m$ perform all-to-all encode for matrix
\begin{equation*}
    \bfG_m'\triangleq \left[ \frac{\begin{array}{ c|c  }
        \bfG_m  & \bfG_{M,m}  \end{array}}{\bfB}\right]
\end{equation*}
using the method introduced in Appendix~\ref{section:app-case-1}, where~$\bfG_{M,m}$ contains the columns of~$\bfG_M$ that corresponds to the processors that distributed to column~$m$ of the grid, and~$\bfB$ may be any matrix that makes~$\bfG'_m$ a square matrix (for an illustrative example, see Figure~\ref{fig:framework-3}).
Together, the communication cost is then
\small
\begin{align*}
    C= C_\text{BR}(\ceil{\textstyle\frac{R}{K}+1},W)
    + \max\left[ {C_\text{A2A}(\bfG'_0),\ldots,C_\text{A2A}(\bfG'_{M-1})}\right].
\end{align*}
\normalsize
\end{appendices}

\end{document}